\PassOptionsToPackage{dvipsnames}{xcolor}

\documentclass[a4paper,twocolumn,unpublished,11pt,accepted=2026-07-06]{quantumarticle}
\pdfoutput=1
\usepackage[normalem]{ulem}
\usepackage[utf8]{inputenc}
\usepackage{amsmath}
\usepackage{amsfonts}
\usepackage{amssymb}
\usepackage{amsthm}
\usepackage{pgfplots}
\usepackage{graphicx}
\usepackage{sidecap} 
\usepackage{cancel}
\usepackage{booktabs}
\usepackage{physics}
\usepackage{mathtools}
\usepackage{dsfont}
\usepackage{comment}
\usepackage{subfig}
\usepackage{ragged2e}
\usepackage[colorlinks]{hyperref}
\hypersetup{
    colorlinks  = true,
    citecolor   = PineGreen,
    linkcolor   = MidnightBlue,
    urlcolor    = PineGreen
}
\usepackage[capitalise]{cleveref}
\crefname{appendix}{Appendix}{Appendices}
\Crefname{appendix}{Appendix}{Appendices}

\graphicspath{{figures/}{./}}

\usepackage{thmtools}
\declaretheorem[parent=section]{theorem}

\declaretheorem[numberlike=theorem, name=Lemma]{lemma}

\declaretheorem[numberlike=theorem, name=Corollary]{corollary}

\usepackage{caption} 
\def\tr{\operatorname{tr}}

\def\ket#1{|#1\rangle}
\def\bra#1{\langle #1|}

\def\c2{\ensuremath{C^{(2)}}}

 \pgfplotsset{compat=1.18}

\newcommand{\del}[1]{}
\newcommand{\add}[1]{#1}

\begin{document}

\title{Reliable Entropy Estimation from Observed Statistics for Device-Independent Quantum Cryptography}

\author{Gereon Koßmann}
\affiliation{Institute for Quantum Information, RWTH Aachen University, Aachen, Germany}
\author{René Schwonnek}
\affiliation{Institut f\"{u}r Theoretische Physik, Leibniz Universit\"{a}t Hannover, Germany}
\begin{abstract}
    We introduce a numerical framework for reliably estimating conditional von Neumann entropies, a central quantity in device-independent quantum cryptography and randomness generation. Our method is based on semidefinite relaxations derived from the Navascués--Pironio--Acín hierarchy and provides entropy bounds directly from observed statistics, assuming only the validity of quantum mechanics. The approach builds on a recent integral representation of entropy and substantially improves efficiency: it requires about half as many support variables as existing methods, while allowing these variables to be chosen projectively, leading to a significant reduction in runtime. These advances enable the practical certification of randomness and security even in noisy conditions, while integrating seamlessly with modern entropy accumulation theorems. Consequently, our framework becomes a versatile tool for quantum cryptographic protocols, broadening the possibilities for secure communication in untrusted environments.
\end{abstract}
\maketitle

\section{Introduction}

Almost a century ago, Max Born \cite{born1926quantenmechanik} famously suggested that the unpredictability observed in individual quantum measurements cannot always be explained by an incomplete knowledge about the experimental setup alone. Instead, this unpredictability can be an intrinsic property of nature itself -- that is, true randomness exists. This idea is strikingly demonstrated by a successful Bell test.

Regarded from the perspective of quantum mechanics, a maximal violation of a Bell inequality like CHSH \cite{Clauser1969} demonstrates two key points. On one hand, it rules out the existence of hidden variables that could secretly determine a outcome. On the other, it shows that the correlated yet individually unpredictable measurement outcomes -- emerging from a pure, maximally entangled state -- are truly random and uncorrelated with any third party not involved in the experiment.

Today, Bell-type correlation experiments that are capable of closing a wide range of loopholes have moved beyond mere thought experiments -- they are now a practical reality (see \cite{Rauch2018,Aspect1982,Schmied2016,Giustina2013,Zhang2022,Hensen2015} and references therein).
However, since all physical experiments must contend with noise and decoherence, a maximal violation indicative of a pure state will most likely remain an idealized model.
This necessitates the task of characterizing and certifying the amount of randomness present in  observed noisy real-world data -- an issue we address in this work.

Recognizing that the outcomes of a quantum measurement can be genuinely random marked not only a fundamental shift in our understanding of the world but also, much later, paved the way for the development of numerous information-theoretic protocols that would be impossible to achieve under the framework of any classical theory \cite{RENNER2008}.
Most prominently, the field of quantum cryptography evolved from this, which by today ranges from the development of near-commercial technologies to foundational research pushing the limits of current capabilities.

\begin{figure}
    \centering
    \vspace{1.5cm}
    \includegraphics[width=0.9\linewidth]{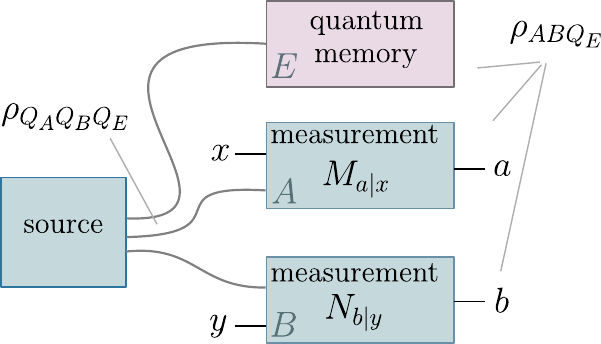}
    \caption{\justifying Basic setting of a quantum mechanical correlation experiment for the generation of randomness or secret key. A source distributes a quantum state $\rho_{Q_A Q_B Q_E}$ between two legit users Alice and Bob, and an adversary Eve. Depending on inputs $x$ and $y$, Alice and Bob perform measurements resulting in classical random variables $A$ and $B$ to which Eve holds quantum side information. Our main focus is to estimate the securely asymptotically extractable randomness of $A$ and/or $B$ based on only considering the observed probability distributions     $p(a,b\vert x,y)$. }
    \label{fig:setting}
\end{figure}

Recent advancements in quantum cryptography, particularly in device-independent quantum key distribution (DIQKD) and randomness extraction (see e.g. \cite{Acn2006,Acn2007} and references therein), harness the intrinsic unpredictability of quantum systems to secure communication against eavesdroppers. DIQKD protocols rely on the principle that correlations observed in measurement outcomes from entangled quantum systems cannot be explained by any local hidden variable theory, as shown by violations of Bell inequalities. This fundamental characteristic implies that the conditional distribution
\begin{align}\label{eq:non_local_distribution}
    p(a,b\vert x,y),
\end{align}
see Fig. \ref{fig:setting},
cannot be reproduced if $a$ and $b$ were deterministic outputs available to an eavesdropper for \del{each}\add{all} input pair $(x,y)$, with values governed by a shared random variable, such as a shared coin toss. Instead, quantum non-locality -- revealed through such correlations -- ensures that outcomes cannot be explained by classical models. This establishes true randomness as a consequence of quantum theory.
The DIQKD approach thus offers an additional layer of security by requiring no trust in the measurement devices themselves. Given this framework, device-independent protocols make the accurate bounding of entropy under non-asymptotic conditions a central challenge \cite{ArnonFriedman2019}. The need to rigorously quantify the randomness of measurement outcomes -- especially in real-world scenarios where noise and finite sample sizes introduce deviations -- underscores the importance of non-asymptotic entropy bounds. Such bounds provide critical insights into the level of security achievable in device-independent settings, where only the observed distribution and the validity of quantum theory are assumed \cite{Dupuis2020,Metger2022} (of course in an actual experiment many more assumptions on the devices come into play; we restrict here on the mathematical framework). As research progresses, these entropy bounds are becoming essential for verifying and implementing DIQKD protocols in practical cryptographic systems, thereby expanding the possibilities for secure communication in untrusted environments \cite{Liu_2021,Liu_2022,Liu_2023}.

In this work, we contribute to the field of device-independent cryptography by introducing a numerical method that provides provable lower bounds on the conditional von Neumann entropy, which quantifies the asymptotic amount of usable randomness extractable from a given experiment. Specifically, we bound the \add{von Neumann} entropy \add{ for a classical-quantum state $\rho_{AQ_E}$, here denoted with $S\left(X\right) \coloneqq -\tr[X \log X]$ with $X$ a positive operator,} 
\begin{align}\label{eq:conditional_von_neumann} H(A \vert Q_E)_{\rho_{AQ_E}} \coloneqq S(\rho_{A Q_E}) - S(\rho_{Q_E}), \end{align}
where $A$ is a classical random variable held by Alice and $Q_E$ is a quantum system held by Eve.
Our method enables certification of randomness using only the observed measurement outcomes from Alice and Bob, under the sole assumption that the experiment is governed by quantum theory. The core idea of our method is to transform a novel class of integral representations of the relative entropy \cite{Frenkel2023} into a sequence of non-commutative optimization problems. These problems can be approximated using the Navascués-Pironio-Acín (NPA) hierarchy of semidefinite programs \cite{Navascus2008,Ligthart_2023}. While the approximation is still computationally demanding for larger systems, our method offers a significant improvement in efficiency compared to previous approaches.

Computing conditional von Neumann entropies in the device\add{-}independent setting has been an outstanding open problem for quite some time \cite{ArnonFriedman2019}. Only recently \del{the first methods for this problem were developed by Tan et al. \cite{Tan_2021} and  by Brown, Fawzi, and Fawzi \cite{Brown2024}, shortly after.}\add{were the first methods for this problem developed, by Tan et al. \cite{Tan_2021} and, shortly after, by Brown, Fawzi, and Fawzi \cite{Brown2024}.} \del{Both are suffering from serious performance issues when it comes to analyse protocols with more than an minimal amount of inputs and outputs on each side. This  limitation is due the fact that both methods require the simultaneous analysis of many non-commutative variables and/or high-degree polynomials.}\add{While both methods can treat scenarios beyond the simplest one (e.g., \cite{Brown2024} handles the $2322$ and tripartite settings), their computational cost grows rapidly with higher accuracy of the approximation. In this work we present a method that is strictly better in the number of non-commuting variables than both of the above, and that further improves on specific performance bottlenecks of these techniques. For instance, the generators $Z$ used for the discretization in \cite{Brown2024} are not projectors, which complicates the resulting relaxation, whereby our construction avoids this issue.}

We have the capability to fix this. In its most economical version, our method can give direct estimates on the von Neumann entropy, whilst requiring minimal computational resources. In detail we need the same amount of resources as in the NPA-based computation of min-entropies \cite{Pironio_2010}, which is a lower but usually suboptimal bound to the von Neumann entropy. Nevertheless, if more resources can be allocated we can enhance the approximation quality and convergence speed to the optimum.

In general, the demand for such numerical tools has increased recently evermore due to significant progress in entropy accumulation theorems \cite{Dupuis2020,Metger2022}, which provide concrete solutions for randomness extraction in finite-size regimes by reducing it to the asymptotic quantity in \eqref{eq:conditional_von_neumann}. Thus, after applying an entropy accumulation theorem, the remaining task is to calculate provable lower bounds for the conditional von-Neumann entropy \eqref{eq:conditional_von_neumann} for specific noise models and input-output data scenarios. \del{Furthermore, by confining the analysis to projective operators, our method achieves a high degree of computational efficiency, making it suitable for real-world applications.}

This paper is structured as follows: In section Sec.~\ref{sec:bounding_von_neumann}, we will present the underlying mathematical framework for our method.  We will introduce and analyse an integral representation of the underlying optimization problem and introduce approximations, that are linear in the underlying state, see Thm.~\ref{thm:main_theorem}. This enables us to formulate the asymptotically extractable randomness as a non-commutative polynomial optimization problem, Thm.\ref{thm:npo_formulation}.
The following sections are dedicated to applications. In Sec.\ref{sec:applications_numerical_examples} we will revisit the well known  2222-setting. Here, Alice has 2 measurements with 2 outcomes, and Bob has 2 measurements with 2 outcomes.  In section Sec.~\ref{sec:beyondchsh} we will focus on settings with more measurements and outcomes. In detail, we will compute the asymptotically extractable randomness for settings in which the CGLMP inequality and the I3322 are measured. For these cases, previous numerical methods do not longer perform well \cite{riveradean2024deviceindependentquantumkeydistribution}. Our method however does, and we demonstrate that direct bounds on the von Neumann entropy certify much higher randomness extraction rates than a comparable ansatz based on the min-entropy. As a last demonstration, presented in section Sec.~\ref{sec:experiment} we will test our method on real world experimental data. The data comes from a test run of the experiment \cite{Zhang2022}, which firstly demonstrated the feasibility of DI-QKD with distant users. We show that using the statistics of more than 2 measurements can actually increase the amount of certifiable randomness when compared to a simple test of the CHSH inequality \cite{Pironio_2010}.    

\section{Bounding the von-Neumann entropy}\label{sec:bounding_von_neumann}
In the following we denote by $\mathcal{B}(\mathcal{H})$ the set of (bounded) linear operators on a Hilbert space $\mathcal{H}$ and $\mathcal{S}(\mathcal{H})$ the set of quantum states on $\mathcal{H}$, i.e. all positive operators with unit trace. The trace on $\mathcal{B}(\mathcal{H})$ is denoted as $\operatorname{tr}[\cdot]$. The positive semidefinite cone is denoted by $\mathcal{P}(\mathcal{H})$. Moreover, any self adjoint operator $A \in \mathcal{B}(\mathcal{H})$, can be uniquely decomposed as a difference $A=A^+-A^-$ of Hilbert-Schmidt orthogonal positive operators $A^+$ and $A^-$. Let $\tr^+[A]\coloneqq\tr[A^+]$ denote the trace of the positive part of $A$ (similarly $\tr^-[A]\coloneqq  \tr[A^-] = \tr^+[-A]$). Note that this is an SDP given by
\begin{equation}\label{eq:variational_tr_plus}
\begin{aligned}
    \operatorname{tr}^+[A] = \sup \ &\tr[PA] \\
    \operatorname{s.t.} \ &0\leq P \leq \mathds{1}.
\end{aligned}
\end{equation}
Our approach for bounding the conditional von-Neumann entropy builds on a recently developed integral representation of the relative entropy\footnote{$\log$ is the logarithm to the base $2$ and $\ln$ the natural logarithm.}
\begin{equation}\label{eq:def_relative_entropy}
\begin{aligned}
    D(\rho \Vert \sigma) \coloneqq  \begin{cases}
        \tr[\rho \log \rho - \rho \log \sigma] & \ker[\sigma] \subseteq \ker[\rho] \\
        + \infty & \text{else}
    \end{cases}
\end{aligned}
\end{equation}
by P. E. Frenkel \cite{Frenkel2023}. For two positive operators $\sigma,\rho \in \mathcal{P}(\mathcal{H})$ acting on a Hilbert space $\mathcal{H}$ the r.h.s. of \eqref{eq:def_relative_entropy} can be expressed by the integration
\begin{align}\label{eq:integral_frenkel}
    D(\rho \Vert \sigma) = \frac{1}{\ln 2}\big(\tr[\rho - \sigma] + \int_{-\infty}^{+\infty} dt \frac{\tr^-[(1-t)\rho + t\sigma]}{\vert t \vert (t-1)^2} \big).
\end{align}
It was shown by Jenčová \cite{Jencova2024}, that the formulation \eqref{eq:integral_frenkel} also holds for infinite-dimensional yet separable Hilbert spaces. It can be further rewritten in a form that is more convenient for our purposes (see \cref{lem:jencova_trace_class})
\begin{equation}\label{eq:integral_jencova}
\begin{aligned}
    D(\rho \Vert \sigma) &= \frac{1}{\ln 2} \big(\tr[\rho - \sigma] + \int_\mu^\lambda \frac{ds}{s} \tr^+[\sigma s- \rho] \\
    &\hspace{1cm} + \tr[\rho] \ln\lambda - (\lambda-1)\tr[\sigma] \big).
\end{aligned}
\end{equation}
Here, the boundaries of the integration range $\mu,\lambda \in \mathbb{R}_{\geq 0}$  have to obey the operator valued inequality
\begin{align}
    \mu \sigma \leq \rho \leq \lambda \sigma.
\end{align}
In the DI settings we want to analyse, the states $\rho$ and $\sigma$ of interest arise from the protocol under consideration, here a generic choice for valid $\lambda$ and $\mu$ can be given.

For the remainder of this work, we assume  tripartite experiments involving Alice, Bob, and Eve, modeled by the axioms of quantum theory as a tripartite state expressed by a positive operator on a Hilbert space with trace equal to $1$,
\begin{align}
    \rho_{Q_A Q_B Q_E} \in \mathcal{S}(\mathcal{H}_{Q_A} \otimes \mathcal{H}_{Q_B} \otimes \mathcal{H}_{QE}).
\end{align}
Importantly, the argument for the device independence specifies nothing concrete about the Hilbert space $\mathcal{H}_{Q_A} \otimes \mathcal{H}_{Q_B} \otimes \mathcal{H}_{Q_E}$. 
In the following we follow the example of generating a raw-key in a DIQKD protocol, even though our methods can be applied for a larger class of problems in \cref{sec:applications_numerical_examples}. Afterwards, we briefly discuss some simple generalizations. 

In \cref{sec:applications_numerical_examples} we elaborate on this assumption for certain QKD scenarios. In a DIQKD protocol we usually have some measurements \add{enumerated by elements from} $\mathcal{X}$ on $\mathcal{S}(\mathcal{H}_{Q_A})$ with finite outcome set $A$ and similarly $\mathcal{Y}$ on $\mathcal{S}(\mathcal{H}_{Q_B})$ with finite outcome set $B$ respectively. Measurements, see Fig.\ref{fig:setting}, are generally understood as sets of POVMs $\{M_{a\vert x}\}_{a \in A}$ for each $x \in \mathcal{X}$ and similarly $\{M_{b\vert y}\}_{b \in B}$ for each $y \in \mathcal{Y}$. From a more abstract perspective a measurement could also be understood as  channel that maps from a commutative algebra into the observable algebra of a quantum system. In the most simplest form of a protocol, the raw-key is now constructed by one specific measurement $\tilde{x} \in \mathcal{X}$. This can be generalized later on \cite{Schwonnek2021}, but we skip the technical details of this for now. The asymptotically extractable randomness of the fixed measurement $\tilde{x}\in \mathcal{X}$ can be described as\footnote{For simplicity we abbreviate $\rho_{AQ_E\vert\tilde{x}} \equiv \rho_{AQ_E\vert X = \tilde{x}}$}
\begin{align}\label{eq:rewriting_conditional_to_von_neumann}
    H(A\vert X = \tilde{x}, Q_E)_{\rho_{A Q_E\vert \tilde{x}}} = - D(\rho_{A Q_E\vert \tilde{x}} \Vert \mathds{1}_A \otimes \rho_{Q_E}).
\end{align}
Hence, we have to lower bound $-D(\rho_{A Q_E\vert \tilde{x}} \Vert \mathds{1}_A \otimes \rho_{Q_E})$. Specifically, device independent lower bounds satisfy the following prototypical type of optimization problem
\begin{equation}\label{eq:prototypical_optimization}
\begin{aligned}
    \inf &-D(\rho_{A Q_E\vert \tilde{x}} \Vert \mathds{1}_A \otimes \rho_{Q_E})& \\
    &\sum_{abxyi} c_{abxyi}\tr[M_{a\vert x}\otimes N_{b\vert y}\rho_{Q_A Q_B}] \geq q_i, \quad 1\leq i \leq m& \\
    &\sum_{a}M_{a\vert x} = \mathds{1}_{A}, \quad x \in \mathcal{X}&\\
    &\sum_{b} N_{b \vert y} = \mathds{1}_B, \quad y \in \mathcal{Y}&\\
    &M_{a\vert x} \geq 0 \quad a\in A, \ x \in \mathcal{X}& \\
    &N_{b\vert y} \geq 0 \quad b \in B, \ y\in \mathcal{Y}& \\
    &\rho_{Q_A Q_B Q_E} \in \mathcal{S}(\mathcal{H}_{Q_A} \otimes \mathcal{H}_{Q_B} \otimes \mathcal{H}_{Q_E})&
\end{aligned}
\end{equation}
whereby we do not aim to specify the concrete underlying Hilbert space. Moreover, the second line can be understood as information about the joint statistics of Alice and Bob, which is nothing more than \eqref{eq:non_local_distribution}. A simple but sufficient and common example would be that the $c_{abxyi} \in \mathbb{R}$ coefficients form a Bell inequality (see e.g. \cite{Brunner2014}). In the following we need to estimate the program \eqref{eq:prototypical_optimization} with an actually quantifiable program. For this purpose, we use the integral representation \eqref{eq:integral_jencova} and derive an upper bound for the relative entropy \eqref{eq:def_relative_entropy} with an operator valued optimization problem. 
\begin{theorem}[Upper bounds for relative entropy]\label{thm:main_theorem}
    Consider two positive operators $\rho,\sigma \in \mathcal{P}(\mathcal{H})$ such that $\rho \leq \lambda \sigma$  for arbitrary $\lambda \in \mathbb{R}_{> 0}$ and some $0<\mu <\lambda$. Then we can upper bound 
    \begin{equation}\label{eq:estimate_main_theorem}
    \int_{0}^\lambda \frac{ds}{s} \tr^+[\sigma s- \rho] \leq \sup_{0\leq P_0,\ldots, P_r \leq \mathds{1}} \sum_{k=0}^r \tr[P_k(\alpha_k\rho + \beta_k \sigma)]
    \end{equation}
    for sets of real numbers 
\begin{widetext}
    \begin{align*}
    \alpha_k \coloneqq \begin{cases}
        -1 & k = 0,\\[4pt]
        -\Big[\Big(1 + \dfrac{t_1}{t_2 - t_1}\Big)\ln \dfrac{t_2}{t_1} - 1\Big] & k = 1,\\[4pt]
        -\Big[ 1 - \dfrac{t_{r-1}}{t_r - t_{r-1}}\ln \dfrac{t_r}{t_{r-1}}\Big] & k = r,\\[6pt]
        -\Big[\Big(1+ \dfrac{t_k}{t_{k+1}-t_k}\Big)\ln  \dfrac{t_{k+1}}{t_k}
          - \dfrac{t_{k-1}}{t_k - t_{k-1}}\ln \dfrac{t_k}{t_{k-1}}\Big] & \text{else},
    \end{cases}
\end{align*}
and
\begin{align*}
    \beta_k \coloneqq \begin{cases}
        \mu & k = 0,\\[4pt]
        \Big[\Big(1 + \dfrac{t_1}{t_2 - t_1}\Big)\ln \dfrac{t_2}{t_1} - 1\Big]\, t_1 & k = 1,\\[4pt]
        \Big[ 1 - \dfrac{t_{r-1}}{t_r - t_{r-1}}\ln \dfrac{t_r}{t_{r-1}}\Big]\, t_r & k = r,\\[6pt]
        \Big[\Big(1+ \dfrac{t_k}{t_{k+1}-t_k}\Big)\ln  \dfrac{t_{k+1}}{t_k}
          - \dfrac{t_{k-1}}{t_k - t_{k-1}}\ln \dfrac{t_k}{t_{k-1}}\Big]\, t_k & \text{else}.
    \end{cases}
\end{align*}
\end{widetext}
    In particular the estimate \eqref{eq:estimate_main_theorem} yields   
\begin{equation}\label{eq:upper_estimate}
\begin{aligned}
    D(\rho \Vert \sigma) &\leq  \frac{1}{\ln 2}(\tr[\rho - \sigma] \\
    &+ \sup_{0\leq P_0,\ldots, P_r \leq \mathds{1}} \sum_{k=0}^r \tr[P_k(\alpha_k\rho + \beta_k \sigma)] \\
    & + \tr[\rho] \ln\lambda - (\lambda-1)\tr[\sigma]).
\end{aligned}
\end{equation}
\end{theorem}
\begin{proof}
    See \cref{appendix:proof_main_theorem}.
\end{proof}
The \cref{thm:main_theorem} provides a sequence of upper bounds that converge to the relative entropy. Specifically, this sequence of optimization problems can be understood as an application of \cite{kossmann2024_optimization}. We emphasize that a finite $\lambda$ is, at least in finite dimensions, equivalent to the statement that the relative entropy is finite. This is a crucial component in these optimization problems, allowing for effective numerical analysis. In the case of conditional von-Neumann entropy, a finite $\lambda$ always exists, as for any bipartite state $\rho_{AB}$ we have the fundamental operator inequality (see, e.g., \cite[Appendix A]{tomamichel2013frameworknonasymptoticquantuminformation}):
\begin{align}
    \rho_{AB} \leq d_A \mathds{1}_A \otimes \rho_B.
\end{align}

Now, let us examine the specific optimization problem in \eqref{eq:prototypical_optimization}. It can be lower-bounded using \cref{thm:main_theorem} as an optimization over operators $P_0, \ldots, P_r$. In particular, we observe that the optimal operators $P_k$, $0 \leq k \leq r$, are projections, as they correspond to projections onto the positive part in the Jordan decomposition of the Hermitian operator $\alpha_k \rho + \beta_k \sigma$. Furthermore, by substituting the specific operators $\rho = \mathds{1}_A \otimes \rho_{Q_E}$ and $\sigma = \rho_{A Q_E \vert \tilde{x}}$, we find that $\lambda = 1$ and $\mu$ can be set to $0$ without loss of generality (see \cref{appendix:proof_main_theorem} for details). Thus, we can rewrite \eqref{eq:prototypical_optimization} as
\begin{equation}\label{eq:optimization_problem_first_rewrite}
\begin{aligned}
    \inf &\inf_{0\leq P_0,\ldots, P_r \leq \mathds{1}} \frac{1}{\ln 2}(\tr[\mathds{1}_A \otimes \rho_{Q_E} - \rho_{AQ_E\vert \tilde{x}}]& \\
    + &\sum_{k=0}^r \tr[P_k(-\alpha_k\rho_{A Q_E \vert \tilde{x}} -\beta_k \mathds{1}_A \otimes \rho_{Q_E})]& \\
   &\sum_{abxyi} c_{abxyi} \tr[M_{a\vert x}\otimes N_{b\vert y}\rho_{Q_A Q_B}] \geq q_i, \quad 1\leq i \leq m& \\
    &\sum_{a}M_{a\vert x} = \mathds{1}_{A}, \quad x \in \mathcal{X}, \quad \sum_{b} N_{b \vert y} = \mathds{1}_B, \quad y \in \mathcal{Y}&\\
    &M_{a\vert x} \geq 0 \quad a\in A, \ x \in \mathcal{X}, \quad N_{b\vert y} \geq 0 \quad b \in B, \ y\in \mathcal{Y}& \\
    &\rho_{Q_A Q_B Q_E} \in \mathcal{S}(\mathcal{H}_{Q_A} \otimes \mathcal{H}_{Q_B} \otimes \mathcal{H}_{Q_E})& \\
    &P_k^2 = P_k \in \mathcal{P}(\mathcal{H}_A \otimes \mathcal{H}_{Q_E}).&
\end{aligned}
\end{equation}
An application of Naimark's dilation theorem additionally allows us to rewrite the state $\rho_{Q_A Q_B Q_E} \in \mathcal{S}(\mathcal{H}_{Q_A} \otimes \mathcal{H}_{Q_B} \otimes \mathcal{H}_{Q_E})$ without loss of generality as a pure state $\psi_{Q_A Q_B Q_E} \in \mathcal{H}_{Q_A}\otimes \mathcal{H}_{Q_B} \otimes \mathcal{H}_{Q_E}$, with all measurements also being projective. By carefully examining the classical-quantum structure of $\rho_{A Q_E \vert \tilde{x}}$, the post-measurement state in our specific key basis $\tilde{x} \in \mathcal{X}$ can then be used to reformulate \eqref{eq:optimization_problem_first_rewrite} as a non-commutative polynomial optimization in \cref{thm:npo_formulation}.

\begin{theorem}[NPO Formulation]\label{thm:npo_formulation}
Calculating lower bounds for the conditional von-Neumann entropy \eqref{eq:conditional_von_neumann} in the device independent setting with knowledge about statistics between Alice and Bobs systems can be reformulated into the following non-commutative polynomial optimization problem
\begin{widetext}
\begin{equation}\label{eq:optimization_problem_second_rewrite}
\begin{aligned}
H(A\vert X = \tilde{x},Q_E) \geq \inf \ &\frac{1}{\ln 2}\Bigg(\vert A\vert - 1 + \sum_{k=0}^r \sum_{a \in A} \bra{\psi_{Q_A Q_B Q_E}} -\alpha_k M_{a\vert x}P_k^{(a)} - \beta_k P_k^{(a)}\ket{\psi_{Q_A Q_B Q_E}}\Bigg)\\
\text{s.t.} \ \ &\sum_{abxyi} c_{abxyi}\bra{\psi_{Q_A Q_B Q_E}}M_{a\vert x} N_{b\vert y}\ket{\psi_{Q_A Q_B Q_E}} \geq q_i, \quad 1\leq i \leq m,\\
&\sum_{a}M_{a\vert x} = \mathds{1}_{A}, \quad x \in \mathcal{X}, \qquad \sum_{b} N_{b \vert y} = \mathds{1}_B, \quad y \in \mathcal{Y},\\
&M_{a\vert x} \geq 0, \quad a\in A, \ x \in \mathcal{X}, \qquad N_{b\vert y} \geq 0, \quad b \in B, \ y\in \mathcal{Y},\\
&[M_{a\vert x},N_{b\vert y}] = [M_{a\vert x},P_k^{(a)}] = [N_{b\vert y},P_k^{(a)}] = 0, \quad a \in A, \ b\in B, \ x \in \mathcal{X}, \ y\in \mathcal{Y}, \ 0 \leq k \leq r,\\
&(P_k^{(a)})^2 = P_k^{(a)}, \quad (P_k^{(a)})^\star = P_k^{(a)}, \quad 1\leq k \leq r, \ a \in A.
\end{aligned}
\end{equation}
\end{widetext}
\end{theorem}
\begin{proof}
    See \cref{appendix:rewriting}
\end{proof}

The \del{polynomial optimization problem} \add{optimization problem}
\eqref{eq:optimization_problem_second_rewrite} \del{can be interpreted as an
optimization over the algebraic constraints within it. In a finite-dimensional
Hilbert space setting, these constraints can be viewed as generating rules of a
$\star$-subalgebra of $\mathcal{B}(\mathcal{H})$, where $\mathcal{H}$ is a
Hilbert space (in infinite dimensions, caution is required to handle the
topological closure of the algebra, though this remains feasible). Assuming that
these relations form an algebra $\mathcal{A}$, we employ the
Navascues-Pironio-Acin (NPA) hierarchy \cite{Navascus2008}, to compute outer
approximations to the optimization problem
\eqref{eq:optimization_problem_second_rewrite}, thus lower bounds.}%
\add{is a \emph{noncommutative polynomial optimization} problem: the objective
is the expectation value of a polynomial in a collection of operators
$X_1,\dots,X_k$, and we optimize over all Hilbert spaces $\mathcal{H}$, all unit
vectors $\ket{\psi}\in\mathcal{H}$, and all operators on $\mathcal{H}$ obeying
the algebraic relations specified by the problem (orthogonality, projector, or
commutation constraints, etc.). Since the dimension of $\mathcal{H}$ is
unrestricted, the problem cannot be solved directly. The
Navascu\'es--Pironio--Ac\'in (NPA) hierarchy \cite{Navascus2008} replaces it by
a sequence of semidefinite programs (SDPs) whose optimal values converge
monotonically to the true optimum, yielding certified lower bounds.}

\add{The construction proceeds as follows. Fix a level $n\in\mathbb{N}$ and let
$\mathcal{S}_n$ be the set of all monomials (products) in the operators $X_i$ of
degree at most $n$, reduced using the algebraic relations of the problem. To a
state $\ket{\psi}$ one associates the \emph{moment matrix} $M^{(n)}$ with rows
and columns indexed by $\mathcal{S}_n$ and entries
\begin{equation}
  M^{(n)}_{u,v} \;=\; \bra{\psi}\, u^\dagger v \,\ket{\psi},
  \qquad u,v\in\mathcal{S}_n.
\end{equation}
For any genuine operator realization, $M^{(n)}$ is Hermitian and positive
semidefinite, its entries satisfy the linear relations inherited from the
algebraic constraints, and the entries corresponding to the objective monomials
are fixed by the problem data. The NPA relaxation at level $n$ optimizes the
objective over all matrices $M^{(n)}\succeq 0$ satisfying these linear
constraints. Dropping the requirement that $M^{(n)}$ actually arise from
operators on a Hilbert space, and keeping only positive semidefiniteness and the
linear constraints, turns this into an SDP.}

\del{The approach involves relaxing the positivity requirement for linear,
continuous functionals on $\mathcal{A}$. Rather than requiring positivity across
the positive cone $\mathcal{A}_+$, we restrict it to a sum-of-squares cone
$\Sigma_2^{(n)} \subset \mathcal{A}_+$ at level $n$. Optimization over this
relaxed SDP can then be viewed as optimizing the dual cone
$(\Sigma_2^{(n)})^\star$, which is representable in SDP form and provides a
convergent hierarchy of outer approximations to
\eqref{eq:optimization_problem_second_rewrite}.}%
\add{Dually, this is a sum-of-squares (SOS) relaxation. A value $\lambda$ is
certified as a lower bound whenever $p-\lambda$ admits a decomposition as a sum
of squares $\sum_i g_i^\dagger g_i$ (modulo the relations) with each $g_i$ a
polynomial of degree at most $n$. Such squares are manifestly positive, so any
$\lambda$ obtained this way is a valid lower bound, and the SDP searches for the
largest one expressible at this level. Increasing $n$ enlarges the set of
admissible squares, so the bounds improve monotonically and converge to the
optimum of \eqref{eq:optimization_problem_second_rewrite}.} For further details
\del{,} we refer to
\cite{Navascus2008,kossmann2023hierarchies,Ligthart_2023} and
\cref{sec:npa_hierarchy}. In our applications\del{,} we use the software
package \del{from}\add{of} \cite{Wittek_2015}\del{ and code is}\add{; the code
is} available \href{https://github.com/gereonkn/DIBounds}{here}.

\begin{figure*}
    \centering
    \subfloat{
            \begin{tikzpicture}
\begin{axis}[width=6cm,height=6cm,scatter/classes={%
    a={mark=o}},
    xlabel = CHSH value,
    ylabel = conditional entropy,
    legend style={at={(0.05,0.95)}, anchor=north west}]
\addplot[only marks,%
    smooth,color=PineGreen]%
table{resultsdi/entropy};
\addlegendentry{lower bound}
\addplot[only marks,%
    scatter src=explicit symbolic,color=MidnightBlue]%
table{resultsdi/exact_entropy};
\addlegendentry{analytic value}

\end{axis}
\end{tikzpicture}
    }      
    \subfloat{
\begin{tikzpicture}
\begin{axis}[width=6cm,height=6cm,scatter/classes={%
    a={mark=o}},
    xlabel = CHSH value,
    ylabel = conditional entropy,
    legend style={at={(0.05,0.95)}, anchor=north west}]
\addplot[only marks,%
    only marks,
  scatter src=explicit symbolic,
  mark=*,
  mark size=1pt,color=PineGreen]%
table{resultsdi/numrel_values.txt};
\addlegendentry{lower bound}
\addplot[
  only marks,
  scatter src=explicit symbolic,
  mark=*,
  mark size=1pt,
  color=MidnightBlue
]
table{resultsdi/numrel_values_colbeck.txt};
\addlegendentry{upper bound \cite{Bhavsar_2023}}

\end{axis}
\end{tikzpicture}
    }
      \subfloat{
            \begin{tikzpicture}
            
\begin{axis}[width=6cm,height=6cm,scatter/classes={%
    a={mark=o}},
    xlabel = white noise,
    ylabel = conditional entropy,
    legend style={at={(0.05,0.95)}, anchor=north west}]
\addplot[only marks,%
    only marks,
  scatter src=explicit symbolic,
  mark=*,
  mark size=1pt,color=PineGreen]%
table{resultsdi/entropy_2322_global};
\addlegendentry{lower bound 2322}

\end{axis}
\end{tikzpicture}

    }
\caption{\justifying
The plots illustrate examples of lower bounds for conditional entropies across different scenarios. In the leftmost plot, we present a proof-of-principle example for lower bounds on the conditional entropy $ H(A \vert X = 0, Q_E)_{\rho} $, depending on various values of the CHSH violation in a 2222-scenario (where each of two parties can choose between two measurements with two outcomes). This example is particularly significant because it allows for a verification of the general correctness of the method; as shown in \cite{Acn2007}, there are analytical values available for this scenario, enabling comparisons with computed results to validate accuracy. The middle plot displays the lower bounds on global randomness, specifically the conditional entropy $ H(AB \vert X = 0, Y = 0, Q_E) $, again in a 2222-scenario and upper bounds developed by \cite{Bhavsar_2023}. Here, the dependency is on the CHSH inequality, and this plot serves to examine the impact of the CHSH violation on the joint conditional entropy of outcomes $ A $ and $ B $ given measurement settings $ X = 0 $ and $ Y = 0 $. Finally, the rightmost plot demonstrates the quantity $ H(AB \vert X = 0, Y = 0, Q_E) $ in a 2322-scenario, where each party has two possible outcomes and Bob has three measurement choices. In this case, we condition on the full distribution of a Werner state depending on dephasing noise. All plots can be generated within seconds to a few minutes of computation on a machine equipped with 12 threads and 128GB of RAM.}
\label{fig:results}
\end{figure*}

\section{Applications and numerical examples}\label{sec:applications_numerical_examples}

In this section, we demonstrate  the methods and tools developed in \cref{sec:bounding_von_neumann}. To structure this discussion, we divide \cref{sec:applications_numerical_examples} into two parts:

In the first part, we examine a CHSH game scenario with two possible inputs and two possible outputs. Our objective here is to compute $ H(A \vert X = 0, Q_E) $, \del{which is a quantity of particular interest in two main areas: one-sided device-independent randomness extraction (DIRE) and device-independent quantum key distribution (DIQKD). In one-sided DIRE, $ H(A \vert X = 0, Q_E) $ provides a measure of the randomness that can be reliably extracted from Alice’s outcomes, assuming a fixed input $ X = 0 $ and some quantum side information $ Q_E $. Similarly, in DIQKD, this conditional entropy helps in determining the asymptotic amount of secure randomness extractable from the quantum correlations between Alice and an adversary.}\add{the relevant quantity for one-sided device-independent randomness extraction (DIRE): it measures the randomness that can be reliably extracted from Alice's outcomes, assuming a fixed input $ X = 0 $ and some quantum side information $ Q_E $. The same rate is also relevant for device-independent quantum key distribution (DIQKD), where it determines the asymptotic amount of secure randomness extractable from the correlations between Alice and an adversary. We note that this fixed-input quantity is the one relevant for spot-checking protocols.}

In the second part, we extend this investigation to the two-sided DIRE task, again within a CHSH game framework. Here, the focus shifts to the quantity $ H(AB \vert X = 0, Y = 0, Q_E) $, where the conditional entropy now includes both Alice’s and Bob’s outcomes. This joint entropy becomes essential for evaluating the global randomness accessible when both parties contribute to the measurement outcomes in the context of device-independent randomness extraction.

For clarity, we abbreviate scenarios based on the number of inputs and outputs per party. Specifically, a scenario with $ n_a $ inputs for Alice, $ n_b $ inputs for Bob, $ o_a $ possible outcomes for each input of Alice, and $ o_b $ possible outcomes for each input of Bob is denoted by $ n_a n_b o_a o_b $. For instance, the CHSH scenario with two inputs and two outputs per party is a 2222-scenario. Furthermore we observe that all quantities under consideration in this work are usable for spot-checking protocols. An interesting direction for future research is to extend this to averaging the entropy terms without spot-checking.

\subsection{one-sided randomness}\label{subsec:one_sided_randomness}

In device-independent quantum information tasks, one-sided (device-independent) randomness is quantified as follows:
\begin{align}\label{eq:one_sided_randomness}
    H(A\vert X  = \tilde{x}, Q_E)_{\rho_{AQ_E\vert \tilde{x}}},
\end{align}
where this quantity captures the uncertainty in the outcome $ A $ when measurement $ X $ takes a fixed value $ \tilde{x} $, conditioned on quantum side information $ Q_E $ held by an eavesdropper. This concept is fundamental in the context of device-independent quantum key distribution (DIQKD) and one-sided randomness extraction \cite{Bhavsar_2023,Devetak_2005}.

In quantum key distribution (QKD), the asymptotic key rate $ r_{\rightarrow} $ can be bounded using the Devetak-Winter formula \cite{Devetak_2005}, which states:
\begin{align}\label{eq:devetak_winter}
    r_{\rightarrow} \geq H(A\vert E)_{\rho_{A Q_E}} - H(A\vert B)_{\rho_{AB}},
\end{align}
where $ H(A\vert E)_{\rho} $ denotes the conditional entropy of $ A $ given the eavesdropper's knowledge $ E $, and $ H(A\vert B)_{\rho} $ denotes the conditional entropy of $ A $ given the legitimate party $ B $. In DIQKD experiments, estimating $ H(A\vert E)_{\rho} $ relies on observed statistics, such as violations of Bell inequalities or even the full probability distribution gathered from the experiment.

This quantity remains significant even in finite-sample-size regimes due to recent advancements in entropy accumulation theorems (EAT) \cite{Dupuis2020,Metger2022}. These theorems provide a framework for quantifying entropy in finite scenarios by demonstrating that finite-size quantities, like the smoothed min-entropy, can be bounded from below by a sum of asymptotic quantities. Thus, an essential part of a DIQKD security proof reduces to reliably lower bounding the asymptotic quantities in  \eqref{eq:devetak_winter}.

Beyond DIQKD,  \eqref{eq:one_sided_randomness} is also important for randomness extraction, a process aimed at generating bits that are random and independent of any other system from a quantum experiment \cite{Liu_2021,Liu_2022,Liu_2023}.

As a concrete example of one-sided randomness extraction, consider the 2222-scenario under the constraint that the underlying, potentially unknown, state $ \rho $ violates the CHSH inequality (Clauser-Horne-Shimony-Holt) with a value of $ \omega \in [2,2\sqrt{2}] $:
\begin{align}\label{eq:CHSH}
    \sum_{x,y=0}^{1}\sum_{a,b=0}^{1} (-1)^{xy+a+b}\,
    \bigl\langle M_{a\vert x}\, N_{b\vert y}\bigr\rangle_{\rho} \geq \omega.
\end{align}
Notably, in 2222-scenarios, \del{the local polytope is uniquely described by the CHSH inequality alone, which is sufficient to certify non-locality in such settings \cite{Brunner2014}.}\add{the only non-trivial facets of the local polytope are given by the CHSH inequality and its symmetries (i.e.\ relabelings of inputs, outputs, and parties). Together with the trivial positivity facets, these fully describe the polytope. As a consequence, violation of the CHSH inequality is sufficient to certify non-locality in such settings \cite{Brunner2014}.} Furthermore, in this specific 2222-scenario, with only the CHSH inequality as a constraint, an analytic solution for  \eqref{eq:one_sided_randomness} is achievable in a device-independent context, as demonstrated by \cite{Acn2007_diqkd}. Therefore, we use this as a benchmark in \cref{fig:results}.

For the experiment shown in \cref{fig:results} in a 2222-scenario with a CHSH inequality as the constraint, we employ a second-level NPA hierarchy with additional monomials of the form $ M_{a\vert x} N_{b\vert y} P^{(a)} $. We also apply the speedups for numerical evaluations discussed in \cref{appendix:numerics}. This modification yields a valid lower bound that serves as a good approximation. In conclusion, the optimization approach is highly efficient, allowing us to partition the interval into many subintervals while keeping the runtime manageable, even on standard personal computers, with computation times remaining within a few seconds.

\subsection{two-sided randomness}

In contrast to the one-sided randomness discussed in the previous section, it is also possible to extract randomness from both parties simultaneously. This approach can significantly enhance the power of quantum random number generators. In such generators, where the experiment is performed within a single laboratory, Alice’s and Bob’s laboratories are effectively combined, allowing randomness to be extracted from both devices in parallel. Consequently, this increases the overall efficiency and robustness of the randomness generation process.

In this two-sided extraction scenario, the key quantity of interest is expressed as:
\begin{align}
    H(AB \vert X = \tilde{x}, Y = \tilde{y}, Q_E).
\end{align}
This expression represents the conditional entropy of the joint outcomes $ A $ and $ B $, given fixed measurement settings $ X = \tilde{x} $ and $ Y = \tilde{y} $, as well as quantum side information $ Q_E $. This joint entropy quantifies the total randomness that can be extracted from the outcomes of both parties under given conditions, making it particularly relevant for applications such as random number generator.

As demonstrated in \cref{appendix:global_distribution}, there is a formula for this two-sided randomness extraction problem that closely resembles the formula presented in \cref{thm:main_theorem}. We formalize this relationship by stating it as a corollary.
 
\begin{corollary}\label{cor:two_sided_randomness}
Bounding the conditional von-Neumann entropy $ H(AB\vert X = \tilde{x}, Y = \tilde{y},Q_E)_{\rho_{ABQ_E\vert \tilde{x},\tilde{y}}}$ from below can be done by solving the following non-commutative polynomial optimization problem
\begin{widetext}
\begin{equation}\label{eq:optimization_two_sided}
\begin{aligned}
    \inf \ \ &\frac{1}{\ln 2}\Bigg(\vert A\vert \cdot \vert B \vert - 1 + \sum_{k=0}^r \sum_{a \in A,\, b \in B} \bra{\psi_{Q_A Q_B Q_E}} -\alpha_k M_{a\vert \tilde{x}}N_{b\vert \tilde{y}}P_k^{(ab)} - \beta_k P_k^{(ab)}\ket{\psi_{Q_A Q_B Q_E}}\Bigg)\\
    \text{s.t.} \ \ &\sum_{abxyi} c_{abxyi}\bra{\psi_{Q_A Q_B Q_E}}M_{a\vert x} N_{b\vert y}\ket{\psi_{Q_A Q_B Q_E}} \geq q_i, \quad 1\leq i \leq m,\\
    &\sum_{a}M_{a\vert x} = \mathds{1}_{A}, \quad x \in \mathcal{X}, \qquad \sum_{b} N_{b \vert y} = \mathds{1}_B, \quad y \in \mathcal{Y},\\
    &M_{a\vert x} \geq 0, \quad a\in A, \ x \in \mathcal{X}, \qquad N_{b\vert y} \geq 0, \quad b \in B, \ y\in \mathcal{Y},\\
    &[M_{a\vert x},N_{b\vert y}] = [M_{a\vert x},P_k^{(ab)}] = [N_{b\vert y},P_k^{(ab)}] = 0, \quad a \in A, \ b\in B, \ x \in \mathcal{X}, \ y\in \mathcal{Y}, \ 0 \leq k \leq r,\\
    &(P_k^{(ab)})^2 = P_k^{(ab)}, \quad (P_k^{(ab)})^\star = P_k^{(ab)}, \quad 1\leq k \leq r, \ a \in A, \ b \in B.
\end{aligned}
\end{equation}
\end{widetext}
\end{corollary}
\begin{proof}
    \cref{appendix:global_distribution}
\end{proof}

We consider two applications of \cref{cor:two_sided_randomness}, each with distinct constraints and setups.

First, we generalize the CHSH-game presented in \cref{subsec:one_sided_randomness} to a two-sided setting. In this case, the setup is governed by the CHSH inequality \eqref{eq:CHSH}, and we examine a 2222-scenario where both Alice and Bob have two measurement inputs, each yielding one of two possible outcomes. The resulting conditional entropy values are similar to those obtained in \cite[Fig. 2]{Brown2024} and are visualized in \cref{fig:results}. To establish upper bounds in this setting, we formulate an optimization problem as described in \cite{Bhavsar_2023}, enabling us to assess lower bounds on randomness asymptotically extractable under CHSH constraints.

A more advanced application involves a 2322-scenario constrained by a whole distribution, inspired by the experimental setup in \cite{Zhang2022}, which originally examined only one-sided randomness. Here, we adapt the setup for two-sided randomness extraction. In this scenario, Alice has two mutually unbiased measurement bases, while Bob has one measurement basis that aligns with Alice's first basis and two additional bases to achieve maximal Bell violation. This configuration allows us to study the randomness generated from both parties simultaneously.

For a concrete example, we consider an honest implementation using an entangled two-qubit state:
\begin{align}
    \ket{\Phi} = \frac{1}{\sqrt{2}} (\ket{00} + \ket{11}).
\end{align}
In this setup, Alice and Bob perform measurements in the $ x $-$ z $ plane with angles specified in \cref{table:angles_Alice_Bob}. Each measurement corresponds to a projection operator $ P_{\alpha} $, defined by
\begin{align}
    P_{\alpha} \coloneqq \frac{1}{2}(\mathds{1} + \sin(\alpha) \sigma_x + \cos(\alpha) \sigma_z),
\end{align}

where $ \alpha $ represents the measurement angle and $ \sigma_x $ and $ \sigma_z $ are Pauli operators. \del{This experimental configuration provides insights into randomness extraction when multiple measurement bases and entangled states are involved.}\add{For the maximally entangled state $\ket{\Phi}$ and these planar measurements, the correlators take the simple form $E(\alpha,\beta) = \cos(\alpha-\beta)$, where $\alpha$ and $\beta$ denote Alice's and Bob's measurement angles. The angles in \cref{table:angles_Alice_Bob} are chosen so that (i) Alice's two bases are mutually unbiased ($\sigma_z$ and $\sigma_x$), and (ii) the resulting correlations violate the CHSH inequality. As discussed above, CHSH is, up to relabeling of inputs, outputs and parties, the only nontrivial Bell inequality in this binary-outcome scenario, so the nonlocality certified here is precisely a CHSH violation.}

\begin{table}[h!]
    \centering
    \begin{tabular}{|c|c|c|}
        \hline
        &Alice & Bob\\
        \hline
        0& $0$ & $\pi/2$ \\
        \hline
        1& $\pi/2$ & $\pi/8$ \\
        \hline
        2& - & $5\pi/4$ \\
        \hline
    \end{tabular}
    \caption{The tables shows the angles for a protocol for randomness extraction in an honest implementation. Importantly, Alice bases are mutually unbiased. }
    \label{table:angles_Alice_Bob}
\end{table}
In the experiment corresponding to \cref{table:angles_Alice_Bob}, randomness extraction is done in the following bases
\begin{align}
    H(AB\vert X = 0,Y = 0, Q_E).
\end{align}
We show in \cref{fig:results} that we match at least the upper bound with $2$ random bits. \del{Moreover, comparing this result with Fig. 3 BFF24 questions whether a third basis on Bob's side has an advantage.}
\add{It is instructive to compare this result with \cite[Fig.~3]{Brown2024}, which plots the two-sided randomness-extraction rate $H(AB\vert X{=}0,Y{=}0,Q_E)$ for the same protocol in the 2222 scenario, with only two measurement bases on Bob's side. Since our rate matches this curve, the third basis on Bob's side provides no advantage for randomness extraction. For device-independent key rates the situation is different. By the Devetak--Winter formula \eqref{eq:devetak_winter}, the rate is given by $H(A\vert X = 0,Q_E) - H(A\vert B)_{\rho_{AB}}$, and the third basis on Bob's side can be chosen to lower the error-correction term $H(A\vert B)_{\rho_{AB}}$, thereby improving the key rate. Hence the additional basis is advantageous for key distribution even though it does not help for pure randomness extraction.}

\add{We note that a closely related protocol was studied in \cite{Brown_2020}, who consider the two-sided, "extended CHSH", randomness from a CHSH-type test and show that it robustly certifies up to two bits of randomness per maximally entangled qubit pair in the low-noise regime. Our analysis of $H(AB\vert X{=}0,Y{=}0,Q_E)$ in this section addresses the same quantity, and our results are consistent with theirs: we likewise recover (at least) the two-bit bound in the honest implementation. This agreement serves as an independent check of our method on a known protocol, while the techniques of \cref{sec:bounding_von_neumann} allow us to treat it with fewer non-commuting variables than previous approaches.}
\section{Beyond CHSH}
\label{sec:beyondchsh}

CHSH has long been the primary Bell inequality used for device-independent randomness. This is mainly because the $2222$-scenarios enjoy the powerful advantage that the $C^\star$-algebra of two projections is well-understood and is essentially as intricate as qubit-based information theory. Many seminal results have been obtained by reducing the device-independent scenario on the algebra of two projections back to the qubit case \cite{Acn2007}. 

However, one can easily conceive of protocols featuring significantly different arrangements of input-output scenarios, which then naturally give rise to different families of Bell inequalities whenever quantum nonlocality arises. In such cases, a straightforward reduction to qubits or even qudits is often not feasible, necessitating either the use of self-testing methods for specific distributions, states and measurements or the acceptance of the full generality of solving the optimization in \cref{thm:npo_formulation}.

In a recent result \cite{riveradean2024deviceindependentquantumkeydistribution}, DIQKD protocols beyond qubits were investigated, proposing key rates for more general scenarios than the CHSH and $2222$-cases. However, computing bounds on these key rates using the Brown--Fawzi--Fawzi (BFF) method can be challenging with standard computational resources. For this reason, the results in \cite{riveradean2024deviceindependentquantumkeydistribution} were obtained by evaluating the min-entropy (i.e., by optimizing the guessing probability). With stronger computational resources, we present a comparison of our method with BFF for I3322 in \cref{fig:comparison_entropy_i3322}, and with BFF for CGLMP in \cref{fig:comparison_entropy_cglmp} and show that our method is on a comparable level faster and gives better bounds. 
\del{Since our focus here is on the fundamental theoretical building blocks of entropy estimation in device-independent scenarios, we concentrate on calculating lower bounds on the asymptotically extractable randomness but in terms of bounds which are converging towards the conditional von Neumann entropy. }
As we will show numerically in the following, 
our tools are able to estimate the randomness for the CGLMP 2233-Bell inequality \cite{Collins_2002} and the I3322 Bell-inequality \cite{Collins_2004} significantly beyond the bounds from the min-entropy.
\begin{equation}\label{eq:min_entropy}
\begin{aligned}
    p_{\operatorname{guess}}&(A\vert E) = &\\
    \sup \ & \sum_{a \in A} \bra{\psi_{Q_A Q_B Q_E}} M_{a\vert \tilde{x}} C_a \ket{\psi_{Q_A Q_B Q_E}}& \\
   &\sum_{abxyi} c_{abxyi}\bra{\psi_{Q_A Q_B Q_E}}M_{a\vert x} N_{b\vert y}\ket{\psi_{Q_A Q_B Q_E}} \geq q_i&\\
   &\quad 1\leq i \leq m& \\
    &\sum_{a}M_{a\vert x} = \mathds{1}_{A}, \quad x \in \mathcal{X}&\\
    &\sum_{b} N_{b \vert y} = \mathds{1}_B, \quad y \in \mathcal{Y}&\\
    &M_{a\vert x} \geq 0 \quad a\in A, \ x \in \mathcal{X}& \\
    &N_{b\vert y} \geq 0 \quad b \in B, \ y\in \mathcal{Y}& \\
    &[M_{a\vert x},N_{b\vert y}] = [M_{a\vert x},C_{a^\prime}] = [N_{b\vert y},C_{a^\prime}] = 0& \\
    &\quad b\in B,\ a,a^\prime \in A, \ x \in \mathcal{X}, \ y\in \mathcal{Y}&\\
    &\sum_a C_a = \mathds{1}, \ C_a \geq 0, \quad a \in A&
\end{aligned}
\end{equation}
and the min entropy is then defined as 
\begin{align}
    H_{\operatorname{min}}(A\vert E) \coloneqq - \log p_{\operatorname{guess}}(A\vert E).
\end{align}
 
\subsection{Randomness from CGLMP}\label{sec:cglmp}
Our first example concerns the CGLMP Bell inequality for a $2233$ scenario \cite[eqs.~(4) and (5)]{Collins_2002}. This is a Bell inequality for two parties, each having two measurement settings and three outcomes per setting. If we define, for generally correlated probability distributions, the expression
\begin{align}
  P(A_a = B_b + k) \coloneqq \sum_{j=0}^{d-1} P\bigl(A_a = j,\; B_b = (j + k)\bmod d\bigr),
\end{align}
\add{where $A_a$ and $B_b$ denote the random variables associated with the outcomes of Alice's and Bob's measurements for settings $a$ and $b$, respectively, taking values in $\mathbb{Z}_d = \{0,\dots,d-1\}$, and arithmetic is understood modulo $d$,}
then the CGLMP Bell inequality can be stated as
\begin{equation}
\begin{aligned}
I_{\operatorname{CGLMP}} \coloneqq \; &\Bigl[P(A_1 = B_1) + P(B_1 = A_2 + 1) + P(A_2 = B_2) \\
&\quad + P(B_2 = A_1)\Bigr] \\
&\; - \Bigl[P(A_1 = B_1 - 1) + P(B_1 = A_2) \\
&\quad + P(A_2 = B_2 - 1) + P(B_2 = A_1 - 1)\Bigr].
\end{aligned}
\end{equation}
To extract randomness, we consider the NPO problem from \eqref{eq:prototypical_optimization} with a single constraint -- namely, the violation of the CGLMP inequality by a specific value. It is straightforward to see that the classical bound is $2$, while an outer bound for the quantum value can be readily computed using the NPA hierarchy, which is approximately $2.91485$. In \cref{fig:results_higher_bell} we present our numerical results. Here, the min-entropy is computed from \eqref{eq:min_entropy} at NPA level $3$, and the bounds from \cref{thm:main_theorem} are obtained at NPA level $2$ by interchanging the sum and the infimum, using a fine grid (around 30 points). The numerical computations can be completed in seconds to minutes. As shown in \cref{fig:results_higher_bell}, the bounds for randomness extracted via \cref{thm:main_theorem} are strictly higher when using the CGLMP inequality.

\begin{figure}[t]
    \subfloat[min-entropy vs. lower bounds from \cref{thm:main_theorem} for CGLMP for 2233 scenario]{
      \begin{tikzpicture}
        \begin{axis}[
            width=7.5cm,
            height=6cm,
            xlabel = {CGLMP value},
            ylabel = {Randomness},
            legend style={at={(0.05,0.95)}, anchor=north west},
            scatter/classes={a={mark=o}}
        ]
          \addplot[
            only marks,
  scatter src=explicit symbolic,
  mark=*,
  mark size=1pt,
            color=PineGreen
          ] table {resultsdi/min_entropy_bounds_cglmp.txt};
          \addlegendentry{min-entropy \eqref{eq:min_entropy}}

          \addplot[
            only marks,
  scatter src=explicit symbolic,
  mark=*,
  mark size=1pt,
            color=MidnightBlue
          ] table {resultsdi/vonNeumann_entropy_bounds_cglmp.txt};
          \addlegendentry{\cref{thm:main_theorem}}

        \end{axis}
      \end{tikzpicture}
    }

    \vspace{1em}

    \subfloat[min-entropy vs. lower bounds from \cref{thm:main_theorem} for I3322 ]{
      \begin{tikzpicture}
        \begin{axis}[
            width=7.5cm,
            height=6cm,
            xlabel = {I3322 value},
            ylabel = {Randomness},
            legend style={at={(0.05,0.95)}, anchor=north west},
            scatter/classes={a={mark=o}}
        ]
          \addplot[
            only marks,
  scatter src=explicit symbolic,
  mark=*,
  mark size=1pt,
            color=PineGreen
          ] table {resultsdi/i3322_guessing.txt};
          \addlegendentry{min-entropy \eqref{eq:min_entropy}}

          \addplot[
            only marks,
  scatter src=explicit symbolic,
  mark=*,
  mark size=1pt,
            color=MidnightBlue,
          ] table {resultsdi/i3322_vN.txt};
          \addlegendentry{\cref{thm:main_theorem}}

        \end{axis}
      \end{tikzpicture}
    }
    \caption{\justifying Lower bounds for the asymptotically extractable randomness $H(A\vert E)$ in a device independent scenario are shown in dependence of the violation of the Bell inequalities from \cite{Collins_2002} and the I3322 Bell inequality. The min-entropy values are obtained by optimizing \eqref{eq:min_entropy} at NPA level $3$, while the direct bounds are obtained with \cref{thm:main_theorem} and an appropriate fine grid. The behaviour of the CGLMP inequality is qualitatively very similar to the well-known CHSH case, which we have considered in \cref{fig:results}. In comparison to that, the I3322 Bell inequality needs a very high violation $\approx 4.83$ (where $4$ is the classical, $\approx 5.0035$ the quantum bound) for extraction of randomness. All results can be achieved on personal devices in the amount of minutes.}
    \label{fig:results_higher_bell}
\end{figure}

\subsection{Randomness from I3322}

In this section we consider the generalization of the \emph{six-state-protocol} originally proposed on qubits in \cite{Bru__1998}. The idea is to consider three PVM's per party with two outcomes and to generate key out of one of the three bases. In \cite{Bru__1998} it is shown that the protocol exhibits slightly improved rates if we consider a device dependent scheme on a qubit. To the best of our knowledge, there are no results for fully device independent lower bounds on the conditional von Neumann entropy in dependence of I3322 violation. We provide in \cref{fig:results_higher_bell} bounds for the asymptotically extractable randomness of the I3322 inequality given by 
\begin{equation}\label{eq:i3322}
\begin{aligned}
    I_{3322} = \ &\langle A_1 B_3\rangle + \langle A_2 B_3\rangle + \langle A_3 B_1\rangle \\
    &+\langle A_3 B_2\rangle +\langle A_1 B_2\rangle +\langle A_2 B_1\rangle - \langle A_2 B_2\rangle \\
    &- \langle A_1 B_1\rangle + \langle A_1\rangle - \langle A_2\rangle + \langle B_1\rangle - \langle B_2\rangle.
\end{aligned}
\end{equation}
It is straightforward to see that the classical bound is $4$, while \del{an outer bound for the quantum value can be readily computed using the NPA hierarchy, which is approximately $5.0035$.}\add{the exact quantum value is not known. In fact, it is conjectured that the maximal quantum violation of $I_{3322}$ is attained only in the limit of infinite-dimensional systems rather than by finite dimensional systems \cite{P_l_2010}. An outer bound on the quantum value can nonetheless be readily computed using the NPA hierarchy and at the level we use this yields approximately $5.0035$, which upper-bounds the true quantum maximum.}

\section{Test on Experimental Data}
\label{sec:experiment}
In this section we \del{report on the combination of our tools with a set of data measured in a real DI-QKD experiment in lab}\add{apply our tools to data measured in a real DI-QKD laboratory experiment}. For this purpose, we \del{performed an analysis of}\add{analyse} data taken in \del{a}\add{an} early test run of the DI-QKD experiment of \cite{Zhang2022}.

\add{We briefly recall the relevant features of that experiment. In \cite{Zhang2022}, event-ready entanglement is generated between two individually trapped single Rydberg atoms located in two laboratories about $400\,$m apart and connected by a $700\,$m optical fibre, realising the random-key-basis DIQKD protocol of \cite{Schwonnek2021}. In each round, one party receives one of four inputs ($x\in\{0,1,2,3\}$) and the other one of two inputs ($y\in\{0,1\}$), and both devices return a binary outcome, so the observed object is the conditional distribution $p(a,b\mid x,y)$ of a $2422$ scenario. The settings with $x=y$ serve as key-generation rounds, while the settings $x\in\{2,3\}$ together with $y\in\{0,1\}$ are used to evaluate the CHSH value. The reference reports $S\approx 2.578$ together with a quantum bit error rate of about $0.078$, and certifies the key from this single Bell value $S$ and the error rate, following \cite{Schwonnek2021}. In our analysis we do not estimate a Bell value first. Instead we feed (part of) the observed behaviour $p(a,b\mid x,y)$ directly into the optimization of \cref{thm:main_theorem}. }

In order to gain insight into the asymptotically extractable randomness directly from the observed statistics, rather than first estimating a Bell inequality and then optimizing the conditional von Neumann entropy \del{with respect to}\add{as a function of} its violation, several \del{steps are required}\add{preprocessing steps are required}. Real experimental data inevitably contain small \add{statistical }errors and are therefore only approximately non-signalling. \del{However, at the heart of the NPA hierarchy lies the assumption that we work within the commuting operator framework, which inherently imposes a form of non-signalling condition.}\add{This is in tension with the NPA hierarchy, since it is formulated in the commuting-operator framework, where every distribution it can certify is necessarily non-signalling.} As a result, \del{simply applying}\add{naively relaxing the data with} an $\varepsilon$-threshold \del{to the data may yield a bound corresponding to a \emph{not}-non-signalling distribution, yet this distribution might still be infeasible at a higher level of the NPA hierarchy}\add{can admit distributions that are not non-signalling, and such distributions may be infeasible at higher levels of the hierarchy, so the resulting bound need not correspond to any physically realizable behaviour}. \del{For this reason, when aiming to exploit}\add{To avoid this while still exploiting} the full informational content of the \del{distribution}\add{data}, we \add{instead }compute the closest non-signalling distribution to the \del{data}\add{measured one} and \del{enforce}\add{impose} it through equality constraints. \add{Projecting raw device-independent data onto the set of physically admissible (in particular non-signalling) distributions before further processing has been considered previously, e.g.\ through maximum-likelihood and least-squares point estimators in \cite{Lin2018}, and our use of the closest non-signalling distribution is in the same spirit.} Since the conditional von Neumann entropy \del{is a continuous function}\add{depends continuously on the distribution}, this \del{approach does introduce a}\add{projection introduces only a small,} controllable error\del{ in the result,} while substantially improving numerical stability.

\del{In this section, we present the numerical results of a 2422-experiment, where we have}\add{We now present numerical results for this $2422$ scenario, for which we have} full access to the distribution $p(a,b \mid x,y)$.
Table~\ref{tab:measurements} summarizes the \del{outcomes}\add{results} for one-sided randomness extraction.
\del{We compare two scenarios for applying Theorem~\ref{thm:main_theorem}: one uses fewer grid points, and the other interchanges the sum and infimum but employs a large number of grid points. Additionally, we examine two data-extraction strategies: one that uses only the CHSH value from the statistics (approximately $2.5958$), and another that utilizes the complete distribution.}\add{The three columns of Table~\ref{tab:measurements} correspond to three configurations of Theorem~\ref{thm:main_theorem}, which differ along two independent axes. The first axis is \emph{which statistics are imposed as constraints}: columns~A and~C impose the full distribution $p(a,b\mid x,y)$, whereas column~B imposes only the statistics of the settings $x\in\{2,3\}$, $y\in\{0,1\}$ that determine the CHSH value (approximately $2.5958$ for these data). The second axis is the \emph{numerical scheme}: A uses $22$ grid points and interchanges the sum and the infimum in Theorem~\ref{thm:main_theorem}, while C uses $7$ grid points without interchanging them. Reading across the table, the comparison A versus B isolates the effect of using the full statistics rather than the CHSH value alone (at a fixed numerical scheme), while A versus C isolates the effect of the numerical scheme (at fixed, full constraints). The comparison A versus B is precisely the question of whether accounting for more of the observed statistics than \cite{Zhang2022} does, which certifies from $S$ alone, yields a quantitative advantage here.} Our results show that, up to a precision of $10^{-4}$, there is no significant difference between using only the CHSH value and using the full statistics\add{ for these data}.
\add{We stress that this does not indicate that the additional statistics are redundant in general. The CGLMP example of Section~\ref{sec:cglmp} exhibits a strict gap. It rather shows that, for the particular behaviour observed in \cite{Zhang2022}, the CHSH value already captures essentially all of the certifiable randomness, so that the more general analysis enabled by our framework coincides with the CHSH-based one to within $10^{-4}$.}
\del{In particular we can compare our results to the theoretical result given by \cite{Pironio_2010} which yields for a CHSH violation of $2.5958$ approximately $0.5759$ bits.}\add{As an analytic benchmark, the bound of \cite{Pironio_2010} yields approximately $0.5759$ bits for a CHSH violation of $2.5958$, in good agreement with the CHSH-only column~B.}
\begin{table}[ht]
    \centering
    \begin{tabular}{l|c|c|c}
    \toprule
         & A & B & C \\
    \midrule
    \multicolumn{4}{c}{\textbf{x values}} \\
    $x = 0$ & 0.3883 & 0 & 0.3915  \\
    $x = 1$ & 0.4178 & 0 & 0.4186 \\
    $x = 2$ & 0.6235 & 0.6235 & 0.6156 \\
    $x = 3$ & 0.5888 & 0.5888 & 0.5779  \\
    \midrule
    \multicolumn{4}{c}{\textbf{y values}} \\
    $y = 0$ & 0.6072 & 0.6069 & 0.5979 \\
    $y = 1$ & 0.6023 & 0.6023 & 0.5921 \\
    \bottomrule
    \end{tabular}
    \caption{\del{In column A we summarize our certified randomness with $22$ grid points in Theorem but sum and infimum are interchanged and we have used the full statistics as equality constraints. In B we just incorporated the constraints regarding $x=2,3$ and $y=0,1$ which would correspond to statistics to calculate the CHSH value. In C we do not interchange sum and infimum, used the full statistic but use only $7$ grid points.}\add{Certified single-party randomness (in bits) for each measurement setting $x$ and $y$ of the $2422$ data of \cite{Zhang2022}, under three configurations of \cref{thm:main_theorem}. \textbf{A}: full statistics imposed as equality constraints, $22$ grid points, sum and infimum interchanged. \textbf{B}: only the constraints on $x\in\{2,3\}$, $y\in\{0,1\}$ that determine the CHSH value, otherwise as in~A, so that the corresponding key settings $x=0,1$ are not constrained, which produces the zero entries. \textbf{C}: full statistics imposed as equality constraints, $7$ grid points, sum and infimum \emph{not} interchanged.} All experiments are\del{ done} at NPA level $2$ and include \del{extra}\add{additional} monomials of the form $M_{a\vert x}N_{b\vert y}P_k^{(a)}$. \del{The results are rounded to a precision $\equiv10^{-4}$.}\add{All values are rounded to a precision of $10^{-4}$.}}
    \label{tab:measurements}
\end{table}

\del{As a last contribution we report on}\add{Finally, we report on} two-sided randomness extraction\del{ whereby we used}\add{. Here we used} $22$ grid points in \cref{thm:main_theorem}\del{, in which the sum and infimum operations were interchanged, and full statistical constraints were employed}\add{ with the sum and infimum interchanged, imposing the full statistics as constraints} at NPA level $2$ with \del{extra}\add{additional} monomials of the form $M_{a\vert x}N_{b\vert y}P_k^{(a)}$. For the measurement settings $x=1$ and $y=0$, the global randomness value was found to be approximately $0.8520$. \del{Comparing this with the results in Figure, yields that it approximately fits exactly what we would expect from CHSH inequality.}\add{Comparing this with \cref{fig:results} shows that it closely matches the value expected from the CHSH inequality.}
\section{Comparison to the method of Brown, Fawzi, and Fawzi}

In order to demonstrate the potential\del{s} our method offers, we will compare it to the method of  Brown, Fawzi, and Fawzi (BFF) \cite{Brown2024}. Conceptually there are some parallels\del{,}\add{.} \add{B}\del{b}oth methods build on lower bounding the von Neumann entropy by upper bounding a relative entropy and both methods use \add{an }integral representation in order to get approximations. The method of Tan, Schwonnek et\add{.} al. \cite{Tan_2021}  follows a different ansatz\del{,}\add{, and} in this work we will leave a  comparison to this method aside. 
In the device-independent setting, the definition of the conditional von Neumann entropy becomes significantly more challenging, since its standard formulation via matrix logarithms cannot be used. At present, there are essentially two different integral representations of the relative entropy and, in the case of a classical conditioning system, the conditional von Neumann entropy:
Kosaki's formula \cite{Kosaki} and Frenkel's formula \cite{Frenkel2023}. Both formulations provide a way to define the conditional von Neumann entropy in a device-independent setting.
Kosaki's formula arises from an interpolation result and can be generalized to all operator monotone functions, whereas Frenkel's formula belongs to the category of $f$-divergences. Since both techniques rely on integration, they can be compared fairly in terms of their efficiency with respect to:
\begin{enumerate}
    \item the quality of the lower bound,
    \item the number of nodes required,
    \item runtime in a competitive example, and
    \item the size of the moment matrices.
\end{enumerate}
\add{Before turning to the examples, we note the structural feature underlying much of our method's advantage. Comparing our optimization problem \eqref{eq:optimization_problem_second_rewrite} with \cite[Eq.~(24)]{Brown2024}, the operators $P_k^{(a)}$ in our formulation play the role of the operators $Z_{a,i}$ in BFF. In contrast to the $Z_{a,i}$, which are not projections, the $P_k^{(a)}$ can be chosen to be projections. This has three immediate consequences. First, it halves the number of additional generators at a given NPA level for the integral discretization, since the NPA hierarchy treats $Z_{a,i}$ and $Z_{a,i}^\star$ as distinct generators, whereas a projection is self-adjoint. Second, any word in the generators in which a $P_k^{(a)}$ appears more than once simplifies automatically. Third, projections have operator norm $1$, so no additional norm constraint is needed, unlike in \cite[Eq.~(24)]{Brown2024}, where $Z_{a,i}Z_{a,i}^\star \leq \alpha_i$ is required. As a result, our program \eqref{eq:optimization_problem_second_rewrite} involves no such operator inequalities and the resulting NPA hierarchy is more stable in practice. In particular, this makes our method strictly better than BFF in the number of non-commuting generators.}
We provide evidence for the comparison with the help of the following examples. The computational experiments were performed on a system equipped with a 13th Gen Intel® Core™ i5-13600KF processor and 128 GB of DDR4 memory operating at a configured speed of 4000 MT/s (maximum rated speed 4800 MT/s).

\subsection{I3322 comparison}

\begin{figure}[ht]
\centering
\begin{tikzpicture}
\begin{axis}[
    width=7.5cm,
    height=6cm,
    xlabel={Bell Score},
    ylabel={Entropy Bound},
    legend style={at={(0.02,0.98)}, anchor=north west},
    grid=major,
    ticklabel style={/pgf/number format/fixed},
    yticklabel style={/pgf/number format/fixed},
]
\addplot+[
            only marks,
  scatter src=explicit symbolic,
  mark=*,
  mark size=1pt,
            mark options={fill=PineGreen, draw=PineGreen}
          ] table {
4.500000000000000000e+00 1.525420625967834193e-07
4.526315789473684070e+00 1.832295812193262072e-07
4.552631578947368141e+00 1.633093761120727730e-08
4.578947368421052211e+00 4.096159006704109750e-08
4.605263157894737169e+00 3.056808554740830166e-07
4.631578947368421240e+00 1.007118107620740284e-07
4.657894736842105310e+00 1.556618648188438670e-07
4.684210526315789380e+00 4.628460843424171199e-07
4.710526315789473450e+00 1.679092266547436808e-07
4.736842105263157521e+00 6.917980592345488899e-08
4.763157894736842479e+00 2.707090776005421314e-07
4.789473684210526550e+00 2.078099505297708110e-07
4.815789473684210620e+00 7.129359866559667470e-07
4.842105263157894690e+00 5.077313933653089528e-02
4.868421052631578760e+00 1.507041882949514966e-01
4.894736842105262831e+00 2.548069908578099452e-01
4.921052631578946901e+00 3.653230771229783458e-01
4.947368421052631859e+00 4.863166164148240433e-01
4.973684210526315930e+00 6.267923059480908776e-01
5.000000000000000000e+00 8.289748950857601661e-01
};
\addlegendentry{BFF Method, 8 notes}

\addplot+[
            only marks,
  scatter src=explicit symbolic,
  mark=*,
  mark size=1pt,
            mark options={fill=MidnightBlue, draw=MidnightBlue}
          ]  table {
4.500000000000000000e+00 5.566142840850529820e-08
4.526315789473684070e+00 6.674287529514537537e-08
4.552631578947368141e+00 4.673817592665566977e-08
4.578947368421052211e+00 3.585837480780858125e-08
4.605263157894737169e+00 2.497735542586209638e-08
4.631578947368421240e+00 2.258575856514707612e-08
4.657894736842105310e+00 5.229514482671883302e-08
4.684210526315789380e+00 4.834278523629390184e-08
4.710526315789473450e+00 8.294289031212347677e-08
4.736842105263157521e+00 1.018457288493961859e-07
4.763157894736842479e+00 2.066620577227681015e-07
4.789473684210526550e+00 3.915809009768076051e-07
4.815789473684210620e+00 2.101530889792948035e-07
4.842105263157894690e+00 5.453100315327730901e-02
4.868421052631578760e+00 1.630775206978070935e-01
4.894736842105262831e+00 2.774730572581514121e-01
4.921052631578946901e+00 4.001103690253762224e-01
4.947368421052631859e+00 5.355189636257428054e-01
4.973684210526315930e+00 6.918083402842137852e-01
5.000000000000000000e+00 8.996078030315310947e-01
};
\addlegendentry{This work, 8 notes}

\addplot+[
            only marks,
  scatter src=explicit symbolic,
  mark=*,
  mark size=1pt,
            mark options={fill=Orchid, draw=Orchid}
          ] table {
4.500000000000000000e+00 9.993443235977412632e-07
4.526315789473684070e+00 1.130905460029337147e-06
4.552631578947368141e+00 1.231645550246640375e-06
4.578947368421052211e+00 9.989124845666951646e-07
4.605263157894737169e+00 1.804081368511528893e-06
4.631578947368421240e+00 1.300789476835312812e-06
4.657894736842105310e+00 6.759637778506802336e-07
4.684210526315789380e+00 6.342034395237668287e-07
4.710526315789473450e+00 6.593301151480519976e-07
4.736842105263157521e+00 7.318030419188035790e-07
4.763157894736842479e+00 8.246937257198417598e-07
4.789473684210526550e+00 2.010547938413272138e-06
4.815789473684210620e+00 3.622032553172892168e-06
4.842105263157894690e+00 5.577336474914464104e-02
4.868421052631578760e+00 1.669470705250624132e-01
4.894736842105262831e+00 2.843067209410877760e-01
4.921052631578946901e+00 4.098533200957313882e-01
4.947368421052631859e+00 5.470615474499163078e-01
4.973684210526315930e+00 7.037549676769438989e-01
5.000000000000000000e+00 9.148099180169972922e-01
};
\addlegendentry{This work, 30 notes}

\end{axis}
\end{tikzpicture}
\caption{Comparison of entropy bounds obtained using the BFF method and the from this work as a function of the score. On this instance 8 notes is already significanly challening in runtime $\sim 10$s per point. As comparison our method with also 8 notes needs only $\sim 0.5$s (see \cref{fig:comparison_runtime_i3322}) and gives higher entropies. We also benchmark our method with 30 notes. The runtime is on the order of $\sim 2$s and we can see that the entropy bound further increases.}
\label{fig:comparison_entropy_i3322}
\end{figure}

We investigate randomness extraction from correlations violating the $I_{3322}$ Bell inequality using two approaches: the Brown--Fawzi--Fawzi (BFF) method with relaxations (1), (3), and (4) from \cite[Rem.~2.6]{Brown2024}, and our proposed technique, for which the analogues of \cite[Rem.~2.6~(3),(4)]{Brown2024} can also be applied. In both methods, we employ $8$ nodes in the quadrature formula for the integration and work at NPA level~2, supplemented with the minimal number of level~3 moment matrix elements required to ensure that the objective functional is well defined. This setup results in moment matrices of size $110 \times 110$ for the BFF method and only $62 \times 62$ for our technique.  

In \cref{fig:comparison_entropy_i3322}, we present the lower bounds obtained from \cref{thm:npo_formulation} and BFF for both methods as a function of the Bell score. We conclude that they are on a comparable scale. In \cref{fig:comparison_runtime_i3322} we provide a runtime comparison: the BFF method requires between approximately $3.58$~s and $18.70$~s per instance over the considered score range, while our method completes in $0.43$--$0.51$~s per instance. This demonstrates that, while both methods produce comparable bounds, our approach achieves a substantial speed-up---often exceeding an order of magnitude---due to the significantly smaller moment matrix size resulting from the simpler objective function and the fact that we can assume the operators in Eve's system to be projective.

\begin{figure}[ht]
\centering
\begin{tikzpicture}
\begin{axis}[
    width=7.5cm,
    height=6cm,
    xlabel={Bell Score},
    ylabel={Runtime (s)},
    legend style={at={(0.02,0.98)}, anchor=north west},
    grid=major,
    ticklabel style={/pgf/number format/fixed},
    yticklabel style={/pgf/number format/fixed},
]
\addplot+[
            only marks,
  scatter src=explicit symbolic,
  mark=*,
  mark size=1pt,
  mark options={fill=PineGreen, draw=PineGreen}
          ] table {
4.500000000000000000e+00 3.581722974777221680e+00
4.526315789473684070e+00 3.599066495895385742e+00
4.552631578947368141e+00 3.567398071289062500e+00
4.578947368421052211e+00 3.573879957199096680e+00
4.605263157894737169e+00 3.521507740020751953e+00
4.631578947368421240e+00 3.586855173110961914e+00
4.657894736842105310e+00 3.665650606155395508e+00
4.684210526315789380e+00 3.680586814880371094e+00
4.710526315789473450e+00 3.691456079483032227e+00
4.736842105263157521e+00 3.790045022964477539e+00
4.763157894736842479e+00 3.767492771148681641e+00
4.789473684210526550e+00 3.876279592514038086e+00
4.815789473684210620e+00 4.098931550979614258e+00
4.842105263157894690e+00 1.063062143325805664e+01
4.868421052631578760e+00 1.231175398826599121e+01
4.894736842105262831e+00 1.327487516403198242e+01
4.921052631578946901e+00 1.353223633766174316e+01
4.947368421052631859e+00 1.496057415008544922e+01
4.973684210526315930e+00 1.662683320045471191e+01
5.000000000000000000e+00 1.869566345214843750e+01
};
\addlegendentry{BFF Method}

\addplot+[
            only marks,
  scatter src=explicit symbolic,
  mark=*,
  mark size=1pt,
            mark options={fill=MidnightBlue, draw=MidnightBlue}
          ]  table {
4.500000000000000000e+00 5.439050197601318359e-01
4.526315789473684070e+00 4.331295490264892578e-01
4.552631578947368141e+00 4.357197284698486328e-01
4.578947368421052211e+00 4.377386569976806641e-01
4.605263157894737169e+00 4.431536197662353516e-01
4.631578947368421240e+00 4.425282478332519531e-01
4.657894736842105310e+00 4.438338279724121094e-01
4.684210526315789380e+00 4.549763202667236328e-01
4.710526315789473450e+00 4.478442668914794922e-01
4.736842105263157521e+00 4.430525302886962891e-01
4.763157894736842479e+00 4.425194263458251953e-01
4.789473684210526550e+00 4.401352405548095703e-01
4.815789473684210620e+00 4.602305889129638672e-01
4.842105263157894690e+00 4.986972808837890625e-01
4.868421052631578760e+00 4.772696495056152344e-01
4.894736842105262831e+00 4.796977043151855469e-01
4.921052631578946901e+00 4.835705757141113281e-01
4.947368421052631859e+00 4.904634952545166016e-01
4.973684210526315930e+00 4.951820373535156250e-01
5.000000000000000000e+00 5.140163898468017578e-01
};
\addlegendentry{This work}

\end{axis}
\end{tikzpicture}
\caption{Runtime comparison of the BFF method and our method (KS) for randomness extraction from the $I_{3322}$ Bell inequality, using $8$ quadrature nodes and NPA level~2 with minimal level~3 extensions.}
\label{fig:comparison_runtime_i3322}
\end{figure}

\subsection{CGLMP comparison}

\begin{figure}[ht]
\centering
\begin{tikzpicture}
\begin{axis}[
    width=7.5cm,
    height=6cm,
    xlabel={Bell Score},
    ylabel={Entropy Bound},
    legend style={at={(0.02,0.98)}, anchor=north west},
    grid=major,
    ticklabel style={/pgf/number format/fixed},
    yticklabel style={/pgf/number format/fixed},
]
\addplot+[
  only marks,
  scatter src=explicit symbolic,
  mark=*,
  mark size=1pt,
  mark options={fill=PineGreen, draw=PineGreen}
] table {
2.000000000000000000e+00 2.417946058498376044e-07
2.101650444444444599e+00 8.673873542137788828e-02
2.203300888888888753e+00 1.910194582520446382e-01
2.304951333333333352e+00 3.089408744284262709e-01
2.406601777777777951e+00 4.407043272696239966e-01
2.508252222222222105e+00 5.883230409721887932e-01
2.609902666666666704e+00 7.561738997464595435e-01
2.711553111111110859e+00 9.528266981321691365e-01
2.813203555555555457e+00 1.199827722737015323e+00
2.914854000000000056e+00 1.569934520548240897e+00
};
\addlegendentry{BFF Method}

\addplot+[
  only marks,
  scatter src=explicit symbolic,
  mark=*,
  mark size=1pt,
  mark options={fill=MidnightBlue, draw=MidnightBlue}
] table {
2.000000000000000000e+00 4.236928583976111177e-08
2.101650444444444599e+00 9.653850093870046389e-02
2.203300888888888753e+00 2.027749096690835418e-01
2.304951333333333352e+00 3.200449133717317429e-01
2.406601777777777951e+00 4.504305261528161464e-01
2.508252222222222105e+00 5.967845668300474360e-01
2.609902666666666704e+00 7.638111176401969349e-01
2.711553111111110859e+00 9.726495914068921111e-01
2.813203555555555457e+00 1.220866429721271462e+00
2.914854000000000056e+00 1.568220996868729955e+00
};
\addlegendentry{This work}

\end{axis}
\end{tikzpicture}
\caption{Comparison of entropy bounds obtained using the BFF method and the method from this work for the CGLMP inequality as a function of the score.}
\label{fig:comparison_entropy_cglmp}
\end{figure}

We investigate randomness extraction from correlations violating the CGLMP Bell inequality using two approaches: the Brown--Fawzi--Fawzi (BFF) method with relaxations (1), (3), and (4) from \cite[Rem.~2.6]{Brown2024}, and our proposed technique, for which the analogues of \cite[Rem.~2.6~(3),(4)]{Brown2024} can also be applied. In both methods, we employ $8$ nodes in the quadrature formula for the integration and work at NPA level~2, supplemented with the minimal number of level~3 moment matrix elements required to ensure that the objective functional is well defined. This setup results in moment matrices of size $230 \times 230$ for the BFF method and only $122 \times 122$ for our technique.  

In \cref{fig:comparison_entropy_cglmp}, we present the lower bounds obtained from \cref{thm:npo_formulation} and BFF for both methods as a function of the Bell score, demonstrating comparable scales. In \cref{fig:comparison_runtime_cglmp} we provide a runtime comparison: across the considered score range, the BFF method requires approximately $69.05$--$242.81$\,s per instance, while our method completes in only $3.08$--$3.67$\,s per instance. This illustrates a substantial speed-up of our approach, attributable to the significantly smaller moment matrices resulting from a simpler objective functional and the ability to assume projective operators on Eve's system.

\begin{figure}[ht]
\centering
\begin{tikzpicture}
\begin{axis}[
    width=7.5cm,
    height=6cm,
    xlabel={Bell Score},
    ylabel={Runtime (s)},
    legend style={at={(0.02,0.98)}, anchor=north west},
    grid=major,
    ticklabel style={/pgf/number format/fixed},
    yticklabel style={/pgf/number format/fixed},
]
\addplot+[
  only marks,
  scatter src=explicit symbolic,
  mark=*,
  mark size=1pt,
  mark options={fill=PineGreen, draw=PineGreen}
] table {
2.000000000000000000e+00 6.905280351638793945e+01
2.101650444444444599e+00 1.271663630008697510e+02
2.203300888888888753e+00 1.599320271015167236e+02
2.304951333333333352e+00 1.793378789424896240e+02
2.406601777777777951e+00 1.937275681495666504e+02
2.508252222222222105e+00 1.927677116394042969e+02
2.609902666666666704e+00 2.173466718196868896e+02
2.711553111111110859e+00 2.218345603942871094e+02
2.813203555555555457e+00 2.428099732398986816e+02
2.914854000000000056e+00 1.085373771190643311e+02
};
\addlegendentry{BFF Method}

\addplot+[
  only marks,
  scatter src=explicit symbolic,
  mark=*,
  mark size=1pt,
  mark options={fill=MidnightBlue, draw=MidnightBlue}
] table {
2.000000000000000000e+00 2.862440347671508789e+00
2.101650444444444599e+00 3.515376806259155273e+00
2.203300888888888753e+00 3.587176322937011719e+00
2.304951333333333352e+00 3.670341014862060547e+00
2.406601777777777951e+00 3.490179777145385742e+00
2.508252222222222105e+00 3.497174501419067383e+00
2.609902666666666704e+00 3.590077877044677734e+00
2.711553111111110859e+00 3.546425342559814453e+00
2.813203555555555457e+00 3.510360956192016602e+00
2.914854000000000056e+00 3.083215951919555664e+00
};
\addlegendentry{This work}

\end{axis}
\end{tikzpicture}
\caption{Runtime comparison of the BFF method and our method (KS) for randomness extraction from the CGLMP inequality, using $8$ quadrature nodes and NPA level~2 with minimal level~3 extensions.}
\label{fig:comparison_runtime_cglmp}
\end{figure}

\section{Summary}

This work focuses on optimizing the conditional von-Neumann entropy within a device-independent framework. From a broader perspective, we ask how much randomness can be extracted from two spatially separated, hypothetical “black boxes” located in Alice’s and Bob’s respective labs, where the outcomes follow the statistical laws of quantum theory. Instead of specifying a concrete quantum experiment represented by a Hilbert space $ \mathcal{H} $ and a trace-class operator $ \rho $ on $ \mathcal{H} $, we rely on the achievable statistics defined by general quantum theory. The core principle underlying the presence of genuine randomness here is the monogamy of correlations, which restricts information sharing in non-signalling theories as demonstrated in \cite{Acn2007}. This principle remains applicable even with only limited knowledge of the output statistics in \eqref{eq:non_local_distribution}. Additionally, bounding the capabilities of Alice, Bob, and Eve to the framework of quantum theory is significant not only for applications but also for advancing our understanding of quantum principles. For instance, \cite{CerveroMartn2025} exploits the gap between general non-signalling theories and quantum theory, highlighting this difference as a powerful tool.

Our approach serves as a versatile framework for deriving lower bounds on the conditional von-Neumann entropy. The main technical contributions of our work include the application of the integral representation from \cite{Frenkel2023} adapted for this setting, as well as the formulation of a non-commutative polynomial optimization problem, which we solve using the NPA hierarchy \cite{Navascus2008}. This approach bridges the gap between two established methods: those in \cite{Brown2024} and \cite{Tan_2021}. The technique from \cite{Brown2024} has become standard in security proofs for device-independent quantum key distribution (DIQKD) and has since been extended to Petz-Rényi divergences \cite{hahn2024boundspetzrenyidivergencesapplications}. Compared to \cite{Brown2024}, our method requires fewer NPA variables, as our projections $ P_k^{(a)} $ in \cref{thm:npo_formulation} reduce complexity relative to the non-Hermitian operators $ Z_{a,i} $ in \cite{Brown2024}. This reduction streamlines the operator set in the sum-of-squares cone and enables us to impose additional constraints within the kernel of the map $ \varphi $ in \cite[Fundamental Lemma]{kossmann2023hierarchies}. Moreover, our method builds on discretization techniques for integration from \cite{kossmann2024_optimization}, while \cite{Brown2024} relies on the Gauss-Radau quadrature method from numerical integration theory. However, the advantages of each integration method for the specific computations required remain open for investigation. The method in \cite{Tan_2021}, though distinct from these approaches, recasts the conditional von-Neumann entropy in \eqref{eq:conditional_von_neumann} as
\begin{align}
    H(A \vert Q_E)_{\rho_{AQ_E}} = S(T_{Q_A \to A}(\rho_{Q_A Q_B})) - S(\rho_{Q_A Q_B}),
\end{align}
interpreting it as an entropy production. Here, the channel $ T_{Q_A \to A} $ acts as the measurement channel on Alice’s system. The method in \cite{Tan_2021} establishes lower bounds through a variational approach, employing the Gibbs variational principle and Golden-Thompson inequalities. A direct comparison of these conceptually related integral-based methods with the entropy production approach would be a valuable avenue for future research.
In \cref{sec:applications_numerical_examples}, we present multiple scenarios to illustrate the effectiveness of our method for randomness extraction. These scenarios include both numerical validations and theoretical comparisons to benchmark results from existing studies.

Firstly, we analyze a $ 2222 $ - scenario (two measurement settings per party, each with two outcomes) under CHSH constraints, which is a standard setting in quantum information theory for testing Bell inequalities. We consider this scenario for both one-sided and two-sided randomness extraction, where randomness is certified by outcomes that are either private to one party (one-sided) or to both parties (two-sided). Our results in this $ 2222 $ - scenario show strong agreement with the analytical values reported in \cite{Acn2007}, demonstrating the accuracy and robustness of our approach. Additionally, we explore a more intricate $ 2322 $ - scenario, where one party has three measurement settings while the other has two. This scenario is used to examine global randomness extraction, which involves jointly certifying randomness across all outcomes. Our findings, \del{obtained under the statistical framework of Werner states,}\add{obtained for an honest implementation in which Alice and Bob share a Werner state $\rho_v = v \ket{\Phi}\bra{\Phi} + (1-v)\,\mathds{1}/4$ with visibility $v$ and measure in the bases that maximize the CHSH violation,} indicate that this more complex setup does not offer a substantial qualitative advantage over the $2222$-scenario. This suggests that increasing the measurement settings in this configuration may not necessarily enhance randomness extraction, a result that merits further investigation. An intriguing direction for future research would be to systematically compare different configurations - such as the $ 2222 $, $ 2322 $ and $ 3322$ - scenarios (where the latter allows each party three measurement settings with two outcomes) - to determine how these variations impact randomness generation under different assumptions of honest implementation or even a random key basis \cite{Schwonnek2021}. This could uncover new insights into the optimal design of randomness extraction protocols. Moreover, it would be valuable to examine randomness extraction in more complex non-local games, such as the magic square game discussed in \cite{CerveroMartn2025}. \del{These games offer alternative frameworks for randomness certification that might provide unique advantages in terms of security and efficiency in quantum protocols.}

\section{Acknowledgements}
GK thanks Mario Berta, Omar Fawzi, Ian George, Edwin Lobo, Stefano Pironio and Lewis Wooltorton for fruitful discussions. GK thanks Marco Tomamichel for hosting him at CQT Singapore in fall 2024. The author would like to thank Rutvij Bhavsar for insightful discussions and for providing access to the data from \cite{Bhavsar_2023}. GK acknowledges support from the Excellence Cluster - Matter and Light for Quantum Computing (ML4Q). GK acknowledges funding by the European Research Council (ERC Grant Agreement No. 948139). 
R.S.\ is supported  by the DFG under Germany's Excellence Strategy - EXC-2123 QuantumFrontiers - 390837967 and  SFB 1227 (DQ-mat) , the Quantum Valley Lower Saxony, and the BMBF projects ATIQ, SEQUIN, QuBRA and CBQD. 
The authors thank Harald Weinfurter and his research group for measuring the test data that was used for benchmarking our methods. The author acknowledge the use of \cite{claude2026,OpenAI_ChatGPT}. 
\section{Author Contribution}
The proofs and the manuscript of the paper has been written by GK. The idea to extend our paper \cite{kossmann2024_optimization} to the DI-setting has been discussed between both authors.

\bibliographystyle{naturemag}
\bibliography{main}
\begin{widetext}
\appendix
\crefalias{section}{appendix}
\crefalias{subsection}{appendix}
\crefalias{subsubsection}{appendix}   

\section{\texorpdfstring{Proof of \cref{thm:main_theorem}}{}}\label{appendix:proof_main_theorem}
The following is a generalization of \cite[Cor. 1]{Jencova2024} to positive trace-class operators.
\begin{lemma}\label{lem:jencova_trace_class}
Let $\rho,\sigma\in\mathcal P(\mathcal H)$ be positive trace-class operators such that there exist $\mu,\lambda\ge 0$ such that $\mu\,\sigma \le \rho \le \lambda\,\sigma$.
Then
\begin{equation}\label{eq:genCor1}
D(\rho\Vert\sigma)
= \tr[\rho -\sigma] + \int_{\mu}^{\lambda} \frac{ds}{s}\,\tr^-\!\big[\rho-s\sigma\big] + \tr[\rho] \ln\lambda - (\lambda-1)\tr[\sigma].
\end{equation}
\end{lemma}

\begin{proof}
We start from the integral representation of relative entropy valid for all positive
trace-class operators (see \cite{Frenkel2023}):
\begin{equation}\label{eq:IF}
D(\rho\Vert\sigma) = \tr[\rho-\sigma]
+ \int_{-\infty}^{\infty} \frac{dt}{\vert t\vert (1-t)^2}\,
\tr^-\!\big[(1-t)\rho+t\sigma\big]\,dt .
\end{equation}
Split the integral at $t=0$ and $t=1$.
For $t\le 0$ we have $1-t>0$ and
\begin{align}
((1-t)\rho+t\sigma)_{-}=(1-t)\bigl(\rho-\tfrac{t}{t-1}\sigma\bigr)_{-}.
\end{align}
With the substitution $s=\tfrac{t}{t-1}$ (which maps $(-\infty,0]\to[0,1]$), one checks that
\begin{align}
\int_{-\infty}^{0}\frac{dt}{\vert t\vert(1-t)^2}\tr^-[(1-t)\rho+t\sigma]\,dt
=\int_{0}^{1}\frac{ds}{s}\,\tr^-[\rho-s\sigma] .
\end{align}
Since $\mu\sigma\le \rho$, the integrand vanishes for $s\in[0,\mu]$, hence the last integral
equals $\int_{\mu}^{1}\frac{ds}{s}\tr^-[\rho-s\sigma]$.

For $t\ge 1$, we use
\begin{align}
((1-t)\rho+t\sigma)_{-}=((t-1)\rho-t\sigma)_{+}=(t-1)\bigl(\rho-\tfrac{t}{t-1}\sigma\bigr)_{+}.
\end{align}
With the same substitution $s=\tfrac{t}{t-1}$ (mapping $[1,\infty)\to[1,\infty)$),
\begin{align}
\int_{1}^{\infty}\frac{dt}{t(1-t)^2}\tr^-[(1-t)\rho+t\sigma]\,dt
= \int_{1}^{\infty}\frac{ds}{s}\,\tr^+[\rho-s\sigma] .
\end{align}
Because $\rho\le \lambda\sigma$, we have $(\rho-s\sigma)_{+}=0$ for all $s\ge \lambda$, so the
integral truncates to $\int_{1}^{\lambda}\frac{ds}{s}\tr^+[\rho-s\sigma]$.

Insert these two pieces in \eqref{eq:IF} to get
\begin{align}
D(\rho\Vert\sigma)=\tr[\rho-\sigma]
+\int_{\mu}^{1}\frac{ds}{s}\tr^-[\rho-s\sigma]
+\int_{1}^{\lambda}\frac{ds}{s}\tr^+[\rho-s\sigma] .
\end{align}
Now use the elementary identity for any self-adjoint $X$:
\(
\tr[X_{+}]=\tr[X_{-}]+\tr[X].
\)
With $X=\rho-s\sigma$ this gives
\begin{align}
\tr^+[\rho-s\sigma]
=\tr^-[\rho-s\sigma] + \tr[\rho-s\sigma] .
\end{align}
Hence
\begin{align}
\int_{1}^{\lambda}\frac{ds}{s}\tr^+[\rho-s\sigma]
=\int_{1}^{\lambda}\frac{ds}{s}\tr^-[\rho-s\sigma]
+\int_{1}^{\lambda}\frac{ds}{s}\bigl(\tr[\rho] - s\,\tr[\sigma]\bigr).
\end{align}
Combining with the $[\mu,1]$ piece collapses the two
$\tr^-[\rho-s\sigma]$-integrals into a single integral over $[\mu,\lambda]$, and the
elementary integral evaluates to
\begin{align}
\int_{1}^{\lambda}\frac{ds}{s}\bigl(\tr[\rho] - s\,\tr[\sigma]\bigr)
=\tr[\rho]\ln\lambda - (\lambda-1)\tr[\sigma].
\end{align}
This yields \eqref{eq:genCor1}.

Finally, independence of the choice of $\mu,\lambda$ follows from the truncation properties:
for $s\le \mu$ we have $(\rho-s\sigma)_{-}=0$, and for $s\ge \lambda$ we have
$(\rho-s\sigma)_{+}=0$. Equivalently, if $\lambda'\ge\lambda$ then on $(\lambda,\lambda']$ one
has $(\rho-s\sigma)_{-}=s\,\tr[\sigma]-\tr[\rho]$, so the change in
$\int\frac{ds}{s}\tr^-[\rho-s\sigma]$ is exactly canceled by the change in
$\tr[\rho]\ln\lambda - (\lambda-1)\tr[\sigma]$, and an analogous statement holds at the lower
limit $\mu$.
\end{proof}

Proof of \cref{thm:main_theorem}:
\begin{proof}
Fix a partition $0<t_1<\cdots<t_r$ with $t_1=\mu$ and $t_r=\lambda$. Following the tools in \cite{kossmann2024_optimization}, set
\begin{align}\label{eq:def_yk}
    y_k := \sup_{0\leq P \leq \mathds{1}} \tr\big[P(\sigma t_k - \rho)\big] .
\end{align}
Then \cite{kossmann2024_optimization} yields the estimate
\begin{equation}\label{eq:appendix_1_proof_1}
\begin{aligned}
    \int_\mu^\lambda \frac{ds}{s}\, \tr^+[\sigma s - \rho]
    &\leq y_1\Big[\Big(1 + \frac{t_1}{t_2 - t_1}\Big)\ln \frac{t_2}{t_1} - 1\Big]
    + y_r\Big[ 1 - \frac{t_{r-1}}{t_r - t_{r-1}}\ln \frac{t_r}{t_{r-1}}\Big] \\
    &\qquad + \sum_{k=2}^{r-1} y_k\Big[\Big(1+ \frac{t_k}{t_{k+1}-t_k}\Big) \ln  \frac{t_{k+1}}{t_k}
      - \frac{t_{k-1}}{t_k - t_{k-1}}\ln \frac{t_k}{t_{k-1}}\Big] .
\end{aligned}
\end{equation}
It is convenient to encode the coefficients by
\begin{align*}
    \alpha_k := \begin{cases}
        -\Big[\Big(1 + \dfrac{t_1}{t_2 - t_1}\Big)\ln \dfrac{t_2}{t_1} - 1\Big] & k = 1,\\[4pt]
        -\Big[ 1 - \dfrac{t_{r-1}}{t_r - t_{r-1}}\ln \dfrac{t_r}{t_{r-1}}\Big] & k = r,\\[6pt]
        -\Big[\Big(1+ \dfrac{t_k}{t_{k+1}-t_k}\Big)\ln  \dfrac{t_{k+1}}{t_k}
          - \dfrac{t_{k-1}}{t_k - t_{k-1}}\ln \dfrac{t_k}{t_{k-1}}\Big] & \text{else},
    \end{cases}
\end{align*}
and
\begin{align*}
    \beta_k := \begin{cases}
        \Big[\Big(1 + \dfrac{t_1}{t_2 - t_1}\Big)\ln \dfrac{t_2}{t_1} - 1\Big]\, t_1 & k = 1,\\[4pt]
        \Big[ 1 - \dfrac{t_{r-1}}{t_r - t_{r-1}}\ln \dfrac{t_r}{t_{r-1}}\Big]\, t_r & k = r,\\[6pt]
        \Big[\Big(1+ \dfrac{t_k}{t_{k+1}-t_k}\Big)\ln  \dfrac{t_{k+1}}{t_k}
          - \dfrac{t_{k-1}}{t_k - t_{k-1}}\ln \dfrac{t_k}{t_{k-1}}\Big]\, t_k & \text{else}.
    \end{cases}
\end{align*}

For the interval $[0,\mu]$, use the convexity of $s\mapsto \tr^+[\sigma s-\rho]$ \cite[Lem.~1]{kossmann2024_optimization} and the fact that $\tr^+[\sigma s-\rho]\to 0$ as $s\to 0$ for positive $\sigma,\rho\in\mathcal P(\mathcal H)$. For $s\in[0,\mu]$,
\[
\tr^+[\sigma s-\rho] \le \Big(1-\frac{s}{\mu}\Big)\tr^+[\sigma\cdot 0-\rho] + \frac{s}{\mu}\tr^+[\sigma\mu-\rho]
= \frac{s}{\mu}\tr^+[\sigma\mu-\rho].
\]
Hence
\begin{equation*}
\begin{aligned}
    \int_0^\mu \frac{ds}{s}\, \tr^+[\sigma s-\rho]
    &\le \tr^+[\sigma\mu-\rho]
     = \sup_{0\le P\le \mathds{1}} \tr\big[P(\sigma\mu-\rho)\big],
\end{aligned}
\end{equation*}
which corresponds to the choice $\alpha_0=-1$ and $\beta_0=\mu$.

Inserting \eqref{eq:def_yk} into \eqref{eq:appendix_1_proof_1}, and using the definitions of $\alpha_k,\beta_k$ for $0\le k\le r$, we obtain with \cref{lem:jencova_trace_class}
\begin{equation*}
\begin{aligned}
    D(\rho\Vert\sigma)
    \le \frac{1}{\ln 2}\Big(\tr[\rho-\sigma]
    + \sup_{0\le P_0,\ldots,P_r\le \mathds{1}}
      \sum_{k=0}^r \tr\big[P_k(\alpha_k\rho+\beta_k\sigma)\big]
    +\tr[\rho] \ln\lambda - (\lambda-1)\tr[\sigma]\Big),
\end{aligned}
\end{equation*}
which is exactly \eqref{eq:upper_estimate}. This proves the claim.
\end{proof}

\section{\texorpdfstring{Proof \cref{thm:npo_formulation}}{}}\label{appendix:rewriting}

\begin{proof}
For simplicity we neglect all the way the whole constraints regadarding Bell-type expressions in \eqref{eq:optimization_problem_second_rewrite}. Those constraints are well-known and came along many times in similar calculations regarding the device independent bounds on Bell violations. We start with the following expression, which is basically the step in \eqref{eq:optimization_problem_first_rewrite}, which is certified by \cref{thm:main_theorem} and $\rho_{Q_A Q_B Q_E}\in \mathcal{S}(\mathcal{H}_{Q_A}\otimes \mathcal{H}_{Q_B} \otimes \mathcal{H}_{Q_E})$
\begin{equation}\label{eq:proof_second_thm_1}
\begin{aligned}
    \inf_{\rho_{Q_A Q_B Q_E}} H(A\vert X = \tilde{x}, Q_E)_{\rho_{AQ_E\vert\tilde{x}}} = &\inf_{\rho_{Q_A Q_B Q_E}} -D(\rho_{A Q_E\vert \tilde{x}} \Vert \mathds{1}_A \otimes \rho_{Q_E})& \\
    \geq &\inf_{\rho_{Q_A Q_B Q_E}} \inf_{0\leq P_0,\ldots, P_r \leq \mathds{1}} \frac{1}{\ln 2}(\tr[\mathds{1}_A \otimes \rho_{Q_E} - \rho_{AQ_E\vert \tilde{x}}]& \\
     & \quad + \sum_{k=0}^r \tr[P_k(-\alpha_k\rho_{A Q_E \vert \tilde{x}} -\beta_k \mathds{1}_A \otimes \rho_{Q_E})].& 
\end{aligned}
\end{equation}
In a next step, we use the fact that 
\begin{align}\label{eq:proof_second_thm_2}
\tr[\mathds{1}_A \otimes \rho_{Q_E} - \rho_{AQ_E\vert \tilde{x}}] = \vert A\vert - 1
\end{align}
and that each of the $P_k$ is a projector in a classical quantum system such that we can write
\begin{align}\label{eq:proof_second_thm_projections}
    P_k = \sum_{a \in A}\ket{a}\bra{a} \otimes P_k^{(a)}.
\end{align}
For some positive operator $P_k^{(a)}$. Calculating $(P_k)^2$ then yields that even the $P_k^{(a)}$ have to be projections if $P_k$ is a projection. In addition we use the following calculation for $\rho_{AQ_E\vert \tilde{x}}$
\begin{equation}\label{eq:proof_second_thm_3}
\begin{aligned}
    \rho_{AQ_E\vert \tilde{x}} &= \sum_{a \in A} p(a\vert x) \ket{a}\bra{a}\otimes \rho_{Q_E\vert A = a, \tilde{x}} \\
    &= \sum_{a \in A} \ket{a}\bra{a} \otimes \tr_{Q_A Q_B}(\rho_{Q_A Q_B Q_E} M_{a\vert \tilde{x}}\otimes \mathds{1}_B).
\end{aligned}
\end{equation}
Before we proceed we briefly argue that of course Eve's marginal does not depend on $\tilde{x}$. This is the non-signalling property of quantum theory or one may calculate it directly
\begin{equation}\label{eq:proof_second_thm_4}
\begin{aligned}
    \rho_{Q_E\vert \tilde{x}} &= \tr_A[\rho_{AQ_E\vert\tilde{x}}] \\
    &= \tr_A[\sum_{a \in A} p(a\vert x) \ket{a}\bra{a}\otimes \rho_{Q_E\vert A = a, \tilde{x}}] \\
    &= \sum_{a \in A} \tr_{Q_A Q_B}[\rho_{Q_A Q_B Q_E} M_{a\vert \tilde{x}} \otimes \mathds{1}_{Q_B}] \\
    &= \tr_{Q_A Q_B}[\rho_{Q_A Q_B Q_E}\sum_{a\in A}M_{a\vert \tilde{x}}\otimes \mathds{1}_{Q_B}] \\
    &= \tr_{Q_A Q_B}[\rho_{Q_A Q_B Q_E} \mathds{1}_{Q_A} \otimes \mathds{1}_{Q_B}] \\
    &= \rho_{Q_E}.
\end{aligned}
\end{equation}
Now we insert \eqref{eq:proof_second_thm_2},\eqref{eq:proof_second_thm_3} and \eqref{eq:proof_second_thm_4} into \eqref{eq:proof_second_thm_1} to get 
    \begin{equation}
        \begin{aligned}
           \ldots &= \frac{1}{\ln 2}\big (\vert A\vert - 1 +  \inf_{\rho_{Q_A Q_B Q_E}} \inf_{P_0^a,\ldots,P_r^a} \sum_{k=0}^r - \alpha_k \tr[\sum_{a^\prime \in A} \ket{a^\prime}\bra{a^\prime} \otimes \tr_{Q_A Q_B}[\rho_{Q_A Q_B Q_E} M_{a\vert \tilde{x}} \otimes \mathds{1}_{Q_B Q_E}] \sum_{a \in A} \ket{a}\bra{a}\otimes P_k^{(a)}] \\
           &\hspace{6cm} + \sum_{k=0}^r -\beta_k \tr[\mathds{1}_A \otimes \rho_{Q_E} \sum_{a \in A} \ket{a}\bra{a}\otimes P_k^{(a)}] \big) \\
           &= \frac{1}{\ln 2}\big (\vert A\vert - 1 +  \inf_{\rho_{Q_A Q_B Q_E}} \inf_{P_0^a,\ldots,P_r^a} \sum_{k=0}^r \sum_{a\in A} - \alpha_k \tr[\ket{a}\bra{a} \otimes \tr_{Q_A Q_B}[\rho_{Q_A Q_B Q_E} M_{a\vert \tilde{x}} \otimes \mathds{1}_{Q_B Q_E}] P_k^{(a)}] \\
           & \hspace{6cm} + \sum_{k=0}^r\sum_{a\in A} -\beta_k \tr[\ket{a}\bra{a} \otimes \rho_{Q_E}P_k^{(a)}]\big) \\
           &= \frac{1}{\ln 2}\big (\vert A\vert - 1 +  \inf_{\rho_{Q_A Q_B Q_E}} \inf_{P_0^a,\ldots,P_r^a} \sum_{k=0}^r \sum_{a\in A} - \alpha_k \tr[\rho_{Q_A Q_B Q_E} M_{a\vert \tilde{x}} \otimes \mathds{1}_{Q_B} \otimes P_k^{(a)}] - \beta_k \tr[\rho_{Q_A Q_B Q_E} \mathds{1}_{Q_A}\otimes \mathds{1}_{Q_B} \otimes P_k^{(a)}]\big).
        \end{aligned}
    \end{equation}

Now using the fact that $\rho_{AQ_E\vert \tilde{x}}\leq \mathds{1}_A \otimes \rho_{Q_E}$ (due to the fact that the state is classical quantum) and thus $\lambda = 1$ yields $\tr [\rho] \ln\lambda - (\lambda-1)\tr[\sigma]= 0$. Moreover, as discussed in \cref{thm:main_theorem} and the proof in \cref{appendix:proof_main_theorem}, $\mu = 0$ can be handled without problems such that the constraints coming from the integral representation in \eqref{eq:integral_jencova} become trivial. Moreover, the operators $P_k^{(a)} \in \mathcal{B}(\mathcal{H}_{Q_E})$ such that they commute with all local measurements on Alice and Bob's side respectively. Adding the fact that they are projections from \eqref{eq:proof_second_thm_projections}, we conclude all the constraints stated in \cref{thm:npo_formulation}. Observing the purity argument of $\rho_{Q_A Q_B Q_E}$ by Naimark's Theorem yields the assertion of the \cref{thm:npo_formulation}, if we replace tensor products with commuting operators.
\end{proof}
\section{\texorpdfstring{Global randomness from \del{the full distribution} \add{two outcomes}}{}}\label{appendix:global_distribution}

We show similarly to \cref{appendix:rewriting} how to get lower bounds for 
$H(AB\vert X = \tilde{x},Y = \tilde{y},Q_E)$ and abbreviate steps similar to \cref{appendix:rewriting}. Similarly, we neglect all Bell-type expressions in \eqref{eq:optimization_problem_second_rewrite}. Applying \cref{thm:main_theorem} for $\rho_{Q_A Q_B Q_E}\in \mathcal{S}(\mathcal{H}_{Q_A}\otimes \mathcal{H}_{Q_B} \otimes \mathcal{H}_{Q_E})$ yields
\begin{equation}\label{eq:proof_second_thm_11}
\begin{aligned}
    \inf_{\rho_{Q_A Q_B Q_E}} H(A B \vert X = \tilde{x}, Y = \tilde{y}  Q_E)_{\rho_{AQ_E\vert\tilde{x}}} = &\inf_{\rho_{Q_A Q_B Q_E}} -D(\rho_{A B Q_E\vert \tilde{x}, \tilde{y}} \Vert \mathds{1}_{AB} \otimes \rho_{Q_E})& \\
    \geq &\inf_{\rho_{Q_A Q_B Q_E}} \inf_{0\leq P_0,\ldots, P_r \leq \mathds{1}} \frac{1}{\ln 2}(\tr[\mathds{1}_{AB} \otimes \rho_{Q_E} - \rho_{ABQ_E\vert \tilde{x},\tilde{y}}]& \\
     & \quad + \sum_{k=0}^r \tr[P_k(-\alpha_k\rho_{A B Q_E \vert \tilde{x},\tilde{y}} -\beta_k \mathds{1}_{AB} \otimes \rho_{Q_E})].& 
\end{aligned}
\end{equation}
Using the fact that 
\begin{align}\label{eq:proof_second_thm_21}
\tr[\mathds{1}_{AB} \otimes \rho_{Q_E} - \rho_{ABQ_E\vert \tilde{x},\tilde{y}}] = \vert A\vert \cdot \vert B \vert - 1
\end{align}
and that each of the $P_k$ is a projector in a classical quantum system, we rewrite
\begin{align}
    P_k = \sum_{a \in A,b \in B}\ket{a,b}\bra{a,b} \otimes P_k^{(ab)}.
\end{align}
with positive operators $P_k^{(ab)} \in \mathcal{P}(\mathcal{H}_E)$. Furthermore, we can apply the following calculations for $\rho_{ABQ_E\vert \tilde{x},\tilde{y}}$
\begin{equation}\label{eq:proof_second_thm_31}
\begin{aligned}
    \rho_{ABQ_E\vert \tilde{x},\tilde{y}} &= \sum_{a \in A, b \in B} p(a,b\vert x,y) \ket{a,b}\bra{a,b}\otimes \rho_{Q_E\vert A = a, B = b, \tilde{x},\tilde{y}} \\
    &= \sum_{a \in A, b \in B} \ket{a,b}\bra{a,b} \otimes \tr_{Q_A Q_B}(\rho_{Q_A Q_B Q_E} M_{a\vert \tilde{x}}\otimes N_{b\vert \tilde{y}}).
\end{aligned}
\end{equation}
Now we insert \eqref{eq:proof_second_thm_21} and \eqref{eq:proof_second_thm_31} into \eqref{eq:proof_second_thm_11} to get 
    \begin{equation}
        \begin{aligned}
           \ldots &= \frac{1}{\ln 2}\big (\vert A\vert \cdot \vert B \vert - 1 +  \inf_{\rho_{Q_A Q_B Q_E}} \inf_{P_0^a,\ldots,P_r^a} \sum_{k=0}^r - \alpha_k \tr[\sum_{a^\prime \in A, b^\prime \in B} \ket{a^\prime,b^\prime }\bra{a^\prime, b^\prime} \otimes \tr_{Q_A Q_B}[\rho_{Q_A Q_B Q_E} M_{a\vert \tilde{x}} \otimes N_{b\vert \tilde{y}}\otimes \mathds{1}_{Q_E}] \\
           &\hspace{3cm} \sum_{a \in A,b\in B} \ket{a,b}\bra{a,b}\otimes P_k^{(ab)}]  + \sum_{k=0}^r -\beta_k \tr[\mathds{1}_{AB} \otimes \rho_{Q_E} \sum_{a \in A,b \in B} \ket{a,b}\bra{a,b}\otimes P_k^{(ab)}] \big) \\
           &= \frac{1}{\ln 2}\big (\vert A\vert \cdot \vert B\vert - 1 +  \inf_{\rho_{Q_A Q_B Q_E}} \inf_{P_0^{ab},\ldots,P_r^{ab}} \sum_{k=0}^r \sum_{a\in A,b \in B} - \alpha_k \tr[\ket{a,b}\bra{a,b} \otimes \tr_{Q_A Q_B}[\rho_{Q_A Q_B Q_E} M_{a\vert \tilde{x}} \otimes N_{b\vert \tilde{y}}\otimes \mathds{1}_{Q_E}] P_k^{(ab)}] \\
           & \hspace{6cm} + \sum_{k=0}^r\sum_{a\in A,b  \in B} -\beta_k \tr[\ket{a,b}\bra{a,b} \otimes \rho_{Q_E}P_k^{(ab)}]\big) \\
           &= \frac{1}{\ln 2}\big (\vert A\vert \cdot \vert B \vert - 1 +  \inf_{\rho_{Q_A Q_B Q_E}} \inf_{P_0^{ab},\ldots,P_r^{ab}} \sum_{k=0}^r \sum_{a\in A,b \in B} - \alpha_k \tr[\rho_{Q_A Q_B Q_E} M_{a\vert \tilde{x}} \otimes N_{b\vert \tilde{y}} \otimes P_k^{(ab)}] \\
           &\hspace{4cm}- \beta_k \tr[\rho_{Q_A Q_B Q_E} \mathds{1}_{Q_A Q_B} \otimes P_k^{(ab)}]\big).
        \end{aligned}
    \end{equation}

Now using the fact that $\rho_{AB Q_E\vert \tilde{x},\tilde{y}}\leq \mathds{1}_{AB} \otimes \rho_{Q_E}$ (due to the fact that the state is classical quantum) and thus $\lambda = 1$ yields $\ln \lambda + 1  - \lambda = 0$. Adding all arguments from \cref{appendix:rewriting} we deduce from \cref{appendix:rewriting} a lower bound on two-sided global randomness 
\begin{align}
    H(AB\vert X=0,Y=0,Q_E) &\geq \frac{1}{\ln 2}\big (\vert A\vert + \vert B \vert - 1 +  \inf_{\rho_{Q_A Q_B Q_E}} \inf_{P_0^a,\ldots,P_r^a} \sum_{k=0}^r \sum_{a\in A, b \in B} - \alpha_k \tr[\rho_{Q_A Q_B Q_E} M_{a\vert \tilde{x}} \otimes N_{b\vert \tilde{y}} \otimes P_k^{(ab)}] \\
    & \hspace{3cm}- \beta_k \tr[\rho_{Q_A Q_B Q_E}  \mathds{1}_{Q_A Q_B} \otimes P_k^{(ab)}]\big).
\end{align}

\section{Algebraic Perspective on the NPA hierarchy}\label{sec:npa_hierarchy}

The formulation presented in \cref{thm:npo_formulation} uses the language of Hilbert spaces and optimization, establishing a framework based on all possible Hilbert space representations for the physical system under consideration. However, for solving the optimization problem more effectively, it is advantageous to recast the problem within the framework of a $ C^\star $-algebra. To this end, we define the complex free $\star$-algebra, denoted $\mathcal{F}(\mathcal{G})$, generated by a set of symbols $\mathcal{G}$ as described in \cite{Ligthart_2023}. In our case, the generating set $\mathcal{G}$ is given by

\begin{equation}
    \mathcal{G} \coloneqq \{P_k^{a},M_{a\vert x},N_{b\vert y} \ \vert \ 0 \leq k \leq r, \ x \in X, \ y \in Y, \ a \in A, \ b \in B\}.
\end{equation}

According to \cite{Ligthart_2023}, we can impose a norm on the algebra $\mathcal{F}(\mathcal{G} \vert \mathcal{R})$ for a set of relations $\mathcal{R}$. This process allows us, via a completion procedure, to obtain a $ C^\star $-algebra denoted $ C^\star(\mathcal{G} \vert \mathcal{R}) $. Lemma 4 in \cite{Ligthart_2023} establishes that there is an equivalence between this algebraic formulation and an optimization over its representations.

Within a $ C^\star $-algebra framework, inner approximations of the positive cone are naturally represented as follows:

\begin{equation}
    \Sigma_2^{(1)} \subset \ldots \subset \Sigma_2^{(n)} \subset \Sigma_2^{(n+1)} \subset \ldots \subset C^\star(\mathcal{G} \vert \mathcal{R})^+.
\end{equation}

Here, the sequence $\{\Sigma_2^{(n)}\}$ represents nested inner approximations of the positive cone $ C^\star(\mathcal{G} \vert \mathcal{R})^+ $. Moving to the dual cones of each $\Sigma_2^{(n)}$, we arrive at a corresponding sequence:

\begin{equation}
    (\Sigma_2^{(1)})^\star \supset \ldots \supset (\Sigma_2^{(n)})^\star \supset (\Sigma_2^{(n+1)})^\star \supset \ldots \supset \mathcal{P}(C^\star(\mathcal{G} \vert \mathcal{R})),
\end{equation}

where $ C^\star(\mathcal{G} \vert \mathcal{R})^+ $ denotes the natural positive cone of the $ C^\star $-algebra $ C^\star(\mathcal{G} \vert \mathcal{R}) $, and $ \mathcal{P}(C^\star(\mathcal{G} \vert \mathcal{R})) $ denotes the space of positive, linear, and continuous functionals on $ C^\star(\mathcal{G} \vert \mathcal{R}) $. By choosing appropriate subsets $\Sigma_2^{(n)}$, we can ensure that the dual cone $ (\Sigma_2^{(n)})^\star $ is representable as a semidefinite program (SDP), facilitating a two-step solution approach.

First, we formulate the initial optimization problem as follows:

\begin{equation}
\begin{aligned}
    c^\star \coloneqq \inf \ &\omega(F_0) \\
    \text{s.t.} \ &\omega(F_i) \leq f_i \quad 1 \leq i \leq m \\
    &\omega \in \mathcal{P}(C^\star(\mathcal{G} \vert \mathcal{R})).
\end{aligned}
\end{equation}

In the first relaxation, we approximate the problem by replacing $\mathcal{P}(C^\star(\mathcal{G} \vert \mathcal{R}))$ with the dual cone $ (\Sigma_2^{(n)})^\star $, yielding

\begin{equation}
\begin{aligned}
    c^\star \geq \inf \ &\omega(F_0) \\
    \text{s.t.} \ &\omega(F_i) \leq f_i \quad 1 \leq i \leq m \\
    &\omega \in (\Sigma_2^{(n)})^\star.
\end{aligned}
\end{equation}

In the second step, we exploit the SDP-representability of $ (\Sigma_2^{(n)})^\star $ within a finite-dimensional matrix algebra, constructed via a set of operators $\{K_j\}$. By mapping each condition $ \{F_i\} $ to specific matrices in this SDP framework, we arrive at the final SDP formulation:

\begin{equation}\label{eq:relaxed_problem}
\begin{aligned}
    \inf \ &\operatorname{tr}(\rho M_{F_0}) \\
    \text{s.t.} \ &\operatorname{tr}(\rho M_{F_i}) \geq f_i \quad 1 \leq i \leq m \\
    &\operatorname{tr}(\rho K_j) = 0 \quad 1 \leq j \leq n \\
    &\rho \geq 0.
\end{aligned}
\end{equation}

An illustrative example of this methodology is the sum-of-squares (SOS) cones, which were initially introduced in \cite{Navascus2008}. For SOS cones, the Python package described in \cite{Wittek_2015} automates the translation of generators and relations into an SDP within a matrix algebra framework.

\section{Additional discussion for the numerics}\label{appendix:numerics}
We repeat in this section the resulting optimization problem from \cref{thm:npo_formulation}
\begin{equation}
\begin{aligned}
H(A\vert X = \tilde{x},Q_E) \geq \inf \ &\frac{1}{\ln 2}(\vert A\vert - 1 + \sum_{k=0}^r \sum_{a \in A} \bra{\psi_{Q_A Q_B Q_E}} -\alpha_k M_{a\vert x}P_k^{(a)} - \beta_k P_k^{(a)}\ket{\psi_{Q_A Q_B Q_E}})& \\
   &\sum_{abxyi} c_{abxyi}\bra{\psi_{Q_A Q_B Q_E}}M_{a\vert x} N_{b\vert y}\ket{\psi_{Q_A Q_B Q_E}} \geq q_i \quad 1\leq i \leq m& \\
    &\sum_{a}M_{a\vert x} = \mathds{1}_{A}, \quad x \in \mathcal{X}&\\
    &\sum_{b} N_{b \vert y} = \mathds{1}_B, \quad y \in \mathcal{Y}&\\
    &M_{a\vert x} \geq 0 \quad a\in A, \ x \in \mathcal{X}& \\
    &N_{b\vert y} \geq 0 \quad b \in B, \ y\in \mathcal{Y}& \\
    &[M_{a\vert x},N_{b\vert y}] = [M_{a\vert x},P_k^{(a)}] = [N_{b\vert y},P_k^{(a)}] = 0 \quad b\in B,\ a \in A, \ x \in \mathcal{X}, \ y\in \mathcal{Y}, \ 0 \leq k \leq r&\\
    &(P_k^{(a)})^2 = P_k^{(a)}, \quad (P_k^{(a)})^\star = P_k^{(a)} \quad 1\leq k \leq r, \ a \in A.&
\end{aligned}
\end{equation}
As discussed in \cref{sec:npa_hierarchy}, the relations 
\begin{equation}
    \begin{aligned}
        &\sum_{a}M_{a\vert x} = \mathds{1}_{A}, \quad x \in \mathcal{X}&\\
    &\sum_{b} N_{b \vert y} = \mathds{1}_B, \quad y \in \mathcal{Y}&\\
    &M_{a\vert x} \geq 0 \quad a\in A, \ x \in \mathcal{X}& \\
    &N_{b\vert y} \geq 0 \quad b \in B, \ y\in \mathcal{Y}& \\
    &[M_{a\vert x},N_{b\vert y}] = [M_{a\vert x},P_k^{(a)}] = [N_{b\vert y},P_k^{(a)}] = 0 \quad b\in B,\ a \in A, \ x \in \mathcal{X}, \ y\in \mathcal{Y}, \ 0 \leq k \leq r&\\
    &(P_k^{(a)})^2 = P_k^{(a)}, \quad (P_k^{(a)})^\star = P_k^{(a)} \quad 1\leq k \leq r, \ a \in A.&
    \end{aligned}
\end{equation}
become equality conditions in \eqref{eq:relaxed_problem} expressed within the $\{K_j\}$. We are now left with the objective function

\begin{align}\label{eq:objective}
    \frac{1}{\ln 2} \left( \vert A\vert - 1 + \sum_{k=0}^r \sum_{a \in A} \bra{\psi_{Q_A Q_B Q_E}} -\alpha_k M_{a\vert x}P_k^{(a)} - \beta_k P_k^{(a)} \ket{\psi_{Q_A Q_B Q_E}} \right),
\end{align}

along with some associated constraints. As discussed in \cite[Rem. 2.6]{Brown2024}, we can implement several optimizations to accelerate the numerical evaluations:

\begin{enumerate}
    \item Modifying the infimum and one of the summations in \eqref{eq:objective} provides lower bounds, reducing the number of required operators $ P_k^{(a)} $ to just $ \Vert A \Vert $. This adjustment significantly accelerates the numerical computations, allowing for the use of a grid with any desired precision.
    
    \item We expand the generating set for the sum-of-squares cone and NPA level $ 2 $ by incorporating additional monomials. In alignment with the approach in \cite{Brown2024}, we include elements of the form $ P_k M_{a\vert x} N_{b\vert y} $, which enriches the algebraic structure and improves convergence properties.
    
    \item Finally, we can set the NPA matrices to be real-valued, which simplifies computations further.

\end{enumerate}

Applying these optimizations enables all computations to complete within seconds.

\end{widetext}
\end{document}